\documentclass{llncs}
\usepackage{hyperref}
\usepackage{proof}
\usepackage{amssymb}
\usepackage{comment}
\begin{document}
\title{A Type Checking Algorithm for Higher-rank, Impredicative and Second-order Types}
\author{Peng Fu}
\institute{Dalhousie University}
           

\maketitle
\begin{abstract}
  We study a type checking algorithm that is able to type check a nontrivial subclass of
  functional programs that use features such as higher-rank, impredicative and
  second-order types. The only place the algorithm requires type annotation is
  before each function declaration. 
  We prove the soundness of the type checking algorithm with
  respect to System $\mathbf{F}_{\omega}$, i.e. if the program is type checked,
  then the type checker will produce a well-typed annotated System $\mathbf{F}_{\omega}$ term.
  We extend the basic algorithm to handle pattern matching and let-bindings.
  We implement a prototype type checker and test it on a variety of functional
  programs.

\end{abstract}

\section{Introduction}
\label{intro}
In the paper \textit{De Bruijn notation as a nested datatype} \cite{bird1999bruijn}, Bird and Paterson defined a version of generalized fold, which has the following type:

\begin{verbatim}
gfoldT :: forall m n b . 
            (forall a . m a -> n a) -> 
            (forall a . n a -> n a -> n a) ->
       	    (forall a . n (Incr a) -> n a) -> 
            (forall a . Incr (m a) -> m (Incr a)) ->
            Term (m b) -> n b  
\end{verbatim}

\noindent Note that the quantified type variables \texttt{n} and \texttt{m} are of the kind \texttt{* -> *}. Moreover, \texttt{Term} and \texttt{Incr} are type constructors of kind \texttt{* -> *}. Although
the type variables \texttt{n} and \texttt{m} have kind \texttt{* -> *}, due to the limitations of
the type inference, they cannot be instantiated
with second-order types such as \texttt{\string\ a . a} or \texttt{\string\ a . String}\footnote{Since it is a type level lambda abstraction, we use \texttt{\string\ a . a} instead of \texttt{\string\ a -> a}.}. As a result, in order to use \texttt{gfoldT} in these situations, one has to duplicate the definition of \texttt{gfoldT} and give it a more specific type. If
we have a type checker that supports a limited form of second-order types, then this kind of code
duplication can be avoided.

A similar situtation also arises when using impredicative types. We know that in principle all the recursive functions can be defined using 
a single fixpoint combinator and pattern matching. But the following definition of \texttt{length2} will not pass the type checker that
does not support impredicative types.

\begin{verbatim}
fix :: forall a . (a -> a) -> a
fix f = f (fix f)  

data Nested :: * -> * where
  NN :: forall a . Nested a
  NCons :: forall a . a -> Nested (List a) -> Nested a

length1 :: forall a . Nested a -> Nat
length1 NN = Z
length1 (NCons x xs) = add one (length1 xs)

length2 :: forall a . Nested a -> Nat
length2 = fix (\ r n -> case n of
                          NN -> Z
                          NCons x xs -> add one (r xs))
\end{verbatim}

\noindent Note that the function \texttt{length1} is counting the number of \texttt{NCons}. This is
an example of polymorphic recursion \cite{mycroft1984}, i.e.
the recursive call of \texttt{length1}
is at a different type. And \texttt{length2} is just the fixpoint representation
of \texttt{length1}. To type
check \texttt{length2}, we would need to instantiate the type variable \texttt{a}
in the type of \texttt{fix} with the type \texttt{forall a . Nested a -> Nat}, which
is a form of impredicative instantiation. Most type checkers do not support this feature
 because it is undecidable in general \cite{wells1999}.
One way to work around this problem is
to duplicate the code for \texttt{fix} and give it a more specific type. 

\begin{verbatim}
fixLength :: ((forall a . Nested a -> Nat) -> 
                  (forall a . Nested a -> Nat)) ->
              (forall a . Nested a -> Nat)
fixLength f = f (fixLength f)  
length2 :: forall a . Nested a -> Nat
length2 = fixLength (\ r n -> ...)
\end{verbatim}

\noindent For any polymorphic
recursive function, we would need this kind of work-around to obtain its fixpoint representation if the type checker does not support impredicative types\footnote{This problem was observed by Peyton-Jones et. al. \cite{jones2007practical}}. 

The goal of this work is to design a type checking algorithm that supports second-order types and impredicative types. One benefit is that it can reduce
the kind of code duplications we just mentioned. The main technical contents of this paper are the
followings.

\begin{itemize}
\item To accomonadate second-order and impredicative types, we use a specialized version of second-order unification based on Dowek's work on linear second-order unification (\cite{dowek1993},\cite{dowek2001}). We called it Dowek's bidrectional matching algorithm (Section~\ref{bidirectional}), it generalizes the first-order unification and the second-order matching. We prove the algorithm is sound and terminating (Appendix \ref{bisound}).

\item Armed with Dowek's bidirectional matching, we describe a type checking algorithm
  inspired by the goal-directed theorem proving and logic programming (\cite{CoqManualV8}, \cite{nilsson1990logic}). We also develop a
  mechanism to handle a subtle scope problem.
  We prove the type checking algorithm is sound with respect to
  System $\mathbf{F}_{\omega}$ \cite{Girard:72} (Section \ref{tcalg}, Appendix \ref{sound:type}). The soundness proof gives rise to a method to generate annotated $\mathbf{F}_{\omega}$ terms from the input programs, which is implemented in a prototype type checker\footnote{The prototype type checker is available at \url{https://github.com/fermat/higher-rank}}.

\item  We extend the basic type checking algorithm to handle pattern matching and let-bindings (Section \ref{extent}). 
    We test the type checker
  on a variety of programs that use higher-rank, impredicative and second-order types,
  these include Bird and Paterson's program~\cite{bird1999bruijn} and
  Stump's impredicative Church-encoded merge sort~\cite{stump2016efficiency} (Appendix \ref{examples}). 
  
\end{itemize}

\section{The main idea and the challenges}
\label{idea}
Consider the following program~\cite{jones2007practical}. Note that we assume the data constructors \texttt{Z :: Nat} and \texttt{True :: Bool}.

\begin{verbatim}
data Pair :: * -> * -> * where
  Pair :: forall a b . a -> b -> Pair a b
poly :: (forall v . v -> v) -> Pair Nat Bool
poly f = Pair (f Z) (f True)
\end{verbatim}

If we use Hindley-Milner algorithm (\cite{hindley1969}, \cite{milner1978}) without using the type annotation 
for \texttt{poly}, we would have to assume the argument of \texttt{f} has an
unknown type \texttt{x}, which eventually leads to a failed unification of
\texttt{Nat -> Nat} and \texttt{Bool -> Bool}.
Instead, we adopt the well-established goal-directed
theorem proving technique founds in most theorem provers (e.g. Coq \cite{CoqManualV8}). To
prove the theorem \texttt{(forall v . v -> v) -> Pair Nat Bool}, we first assume \texttt{f :: forall v . v -> v}, then we just need to show \texttt{Pair (f Z) (f True) :: Pair Nat Bool}. We
now apply \texttt{Pair :: forall a b . a -> b -> Pair a b} to the goal \texttt{Pair Nat Bool}, this resolves to two subgoals \texttt{f Z :: Nat} and \texttt{f True :: Bool}. We know these two subgoals holds because we have \texttt{f :: forall v . v -> v}. 

In general, to type check a program $\lambda x . e$ with a type $\forall a . T_1 \to T_2$,
we will type check $e : T_2$ assuming $x : T_1$, where the type variable $a$ in $T_1$ and $T_2$
behaves as a constant (called \textit{eigenvariable}). To type check a program $f \ e_1\ ...\ e_n$
with the type $T$, where $f : \forall a . T_1 \to\ ... \ \to T_n \to T'$,
we will first unify $T'$ with $T$ (the type variable $a$ in $T_1,.., T_n, T'$ behaves as \textit{free variable}), obtaining a unifier $\sigma$. Then we will type check $e_i : \sigma T_i$ for $1 \leq i \leq n$. Notice the different behaviors of the quantified type variable $a$ in
the two cases. When a type variable is introduced as an eigenvariable, we call the introduction
\textit{type abstraction}, when a type variable is introduced as a free variable, we call
the introduction \textit{type instantiation}. Although this idea of type checking
works perfectly for the \texttt{poly} example, it is not obvious how it
can be scale to a more general setting. Indeed, we will need to address the following problems.

\begin{itemize}
 

\item \textbf{Finding an adequate notion of unification}. Impredicative polymorphism means
  that a type variable can be instantiated with any type (which includes forall-quantified types). Consider the following
  program.

\begin{verbatim}
data List :: * -> * where
  Nil :: forall a . List a 
  Cons :: forall a . a -> List a -> List a
test :: List (forall a . List a)
test = Nil
\end{verbatim}

To type check \texttt{test}, we need to unify \texttt{List a} and \texttt{List (forall a . List a)},  which is beyond first-order unification as the forall-quantifed type \texttt{forall a . List a} is not a first-order type. 
First-order unification can not work with second-order types neither, consider the following program.

\begin{verbatim}
data Bot :: * where 
data Top :: * where
k1 :: forall p . p Bot -> p Top
k1 = undefined
k2 :: forall p . p Top -> p Top
k2 = undefined
a1 :: Bot -> Top
a1 = k1
a2 :: Top -> Top
a2 = k2  
\end{verbatim}

Note that the type variable \texttt{p} in \texttt{k1, k2} is of the kind \texttt{* -> *}. We should be able to type check \texttt{a1} by instantiating the type variable \texttt{p} in \texttt{k1} with
the type identity \texttt{\string\ a . a}. This would require unifying
\texttt{(p Bot)} and \texttt{Bot}, which is an instance of undecidable second-order unification~\cite{goldfarb1981}. Besides the problem of undecidability, second-order types also raises a concern of type ambiguity. For
example, to type check \texttt{a2}, we can again instantiate the type variable \texttt{p} in \texttt{k2} with the type identity \texttt{\string\ a . a}, but nothing prevents us to instantiate \texttt{p} with the type constant function \texttt{\string\ a . Top}. Thus there
can be two different type annotations (derivations) for \texttt{a2}. 

\textbf{Our approach}. 
Following the usual practice in higher-order unification \cite{dowek2001}, the unifier of $\forall a .T$ and $\forall b. T'$ is the
unifier $\sigma$ of $[c/a]T$ and $[c/b]T'$, provided the variable $c$ is a fresh eigenvariable and $c$ does not appears in the codomain of $\sigma$.
To handle second-order types, we use a decidable version of second-order unification due
to Dowek~\cite{dowek1993}, it generalizes first-order unification and second-order matching.



\textbf{Drawback}. 
The unification algorithm we use could generate multiple (finitely many) unifers when there are second-order type variables.
This implies that there may be multiple
successful typing derivations for a program when it uses second-order type variables.
For the purpose of type checking, it is enough to pick the first successful derivation because all the typing annotations will be erased when we run the program. If all the derivations
fail, then the type checking fails. So second-order types will introduce a kind of nondeterminism during type checking.


\item \textbf{Handling type abstraction}. We know that it is safe to perform type abstraction 
  when we
  are defining a polymorphic function that has at least one input. For the other situations, it is not straightforward to decide at which point to perform type abstraction.
  A common decision is always perform type abstraction for the outermost
  forall-quantified variables. Consider the following program.

\begin{verbatim}
data F :: * -> * where
fix :: forall a . (a -> a) -> a
fix f = f (fix f)  
l :: forall x . F x -> F x
l = undefined
l' :: (forall x . F x) -> (forall x . F x)
l' = undefined 
test1 :: forall y . F y
test1 = fix l
test2 :: forall y . F y
test2 = fix l'
\end{verbatim}

The program \texttt{test1} can be type checked by first abstracting the outermost variable \texttt{y}, then we need to type check \texttt{fix l} with the type \texttt{F y} (with \texttt{y} as an eigenvariable). This is the case because we can
 instantiate the type variable \texttt{a} in the type of \texttt{fix} with \texttt{F y}, and instantiate the quantified variable \texttt{x} with \texttt{y} in the type of \texttt{l}. 
On the other hand, to type check the program \texttt{test2}, 
we must not perform type abstraction. 

 \textbf{Our approach}. To type check both \texttt{test1} and \texttt{test2},
 we decide to branch the type checking when checking an application (which includes single program variable or constructor) against
 a forall-quantified type. Our type checker always performs type abstraction when checking a polymorphic function that has at least one input.
 For example, when checking program such as \texttt{f x .. = e} with the type \texttt{f :: forall a . T}, then we would
 abstract the outermost type variable \texttt{a}. But when 
 checking a application against a polymorphic type, the type checker
 will make two branches, in one branch the type checker will perform type abstraction
 and in the other the type checker does not. 
 For example, when checking
 \texttt{f g} with the type \texttt{forall a . T}, we would check
 both \texttt{f g} against \texttt{forall a . T} and \texttt{f g} against \texttt{T}. 

 \textbf{Drawback}. Our decision on checking an application against a forall-quantified
 type also introduce nondeterminism. When checking
 single program variable or constructor against a polymorphic type, branching
 is at no cost as these are just two additional leaves. But in the other cases
 branching does mean the type checker will do extra work. 
 

\item \textbf{Scope management}.
Consider the following program.

\begin{verbatim}
k1 :: forall q . (forall y . q -> y) -> Bot
k1 = undefined
k2 :: forall x . x -> x
k2 = undefined
test :: Bot
test = k1 k2
\end{verbatim}

The program \texttt{test} appears to be well-typed, as we can instantiate the
variable \texttt{q} in the type of \texttt{k1} with \texttt{y}, then we can apply
\texttt{k1} to \texttt{k2}. But this is not the case because 
\texttt{q} is incorrectly referred to the bound variable \texttt{y}. 
When using our algorithm to check \texttt{k2} against \texttt{forall y . q -> y} ($\texttt{q}$ is a free variable),
in one branch the algorithm will try to unify \texttt{x -> x} with \texttt{forall y . q -> y}, which fails. In another branch, the algorithm will perform
type abstraction, i.e. it will check \texttt{k2} against the type \texttt{q -> y} (\texttt{y}
is an eigenvariable). Since \texttt{x -> x} unifies with \texttt{q -> y} (the unifier is $[\texttt{y}/\texttt{q}, \texttt{y}/\texttt{x}]$), without proper scope management, our algorithm will wrongly report the success on the second branch.  


\textbf{Our approach}. To handle 
the scope problem, we introduce a notion of \textit{scope value} for variables. Informally, when each
variable (free variable or eigenvariable) is first introduced, it will be assigned a scope value. A variable introduced later
will have a scope value larger than a variable introduced earlier. 
When a free variable $a$ is substituted by a type $T$,
we require all the eigenvariables in $T$ to have a smaller scope value compared to $a$'s, i.e. $a$ can only refer to the eigenvariables that are introduced before $a$. So in our example,
when type checking \texttt{test}, the scope value for the free variable \texttt{q} will be $1$
and the scope value for the eigenvariable \texttt{y} will be $2$, which is larger than $1$, hence the substitution $[\texttt{y}/\texttt{q}]$ gives rise to
a scope error. We incorporate a scope checking process into the type checking algorithm, which
is essential for the soundness of the type checking. 

\textbf{Drawback}. When a free variable $a$ is substituted by a type $T$,
other than eigenvariables and constants, $T$ may contain free variables.
The question now is what if these free variables have scope values larger than $a$'s.
For example, suppose the scope value for $a$ is $3$, but $T$ contains
a free variable $b$ with scope value $5$. We allow such substitution, but we need to update
the scope value of $b$ to the smaller value $3$, this is to prevent $b$ (and $a$ indirectly) later refer to any eigenvariable
with the scope value $4$. Thus when a unifier is generated, we need to perform
scope value check as well as updating the scope values. This complicates the presentation of the type checking
algorithm, but we manage to prove that the scope checking and updating ensures soundness.
\end{itemize}

\section{A type checking algorithm for impredicative and second-order types}
We describe a type checking algorithm for higher-rank, impredicative and second-order types in
this section. Higher-rank types means the forall-quantifiers can appear anywhere in a type.
Impredicative types means type variables can be instantiated with any types (includes the forall-quantified types). Second-order types means type variables of kind $* \to\ ... \ \to * \to *$ can be instantiated with the lambda-abstracted types. All of these features are available in System $\mathbf{F}_{\omega}$,
which will be the target language for our type checking algorithm. Note that the type checking problem for
System $\mathbf{F}_{\omega}$ with annotations is decidable. We use the
terminology \textit{proof checking} to mean checking $\mathbf{F}_{\omega}$ with annotations,
and we use the terminology \textit{type checking} to mean giving a type $T$ and a unannotated term $e$, construct an annotated term $p$ in $\mathbf{F}_{\omega}$ such that it can
be proof checked with type $T$ and $p$ can be erased to $e$. Thus our type checking algorithm will always produce an
annotated term if the type checking is successful.

\begin{figure}

  {\footnotesize
  \textit{Annotated Expressions} $p ::= c\ | \ x \ | \ p \ p' \ | \ \lambda x : T . p \ | \ \lambda a . p \ | \ p \ T$

  \textit{Unannotated Expressions} $e ::= c\ | \ x \ | \ e \ e' \ | \ \lambda x . e$
  
  \textit{Types} $T ::= \ C \ | \ a \ |\ \forall a . T \ |\ T \to T' \ | \ T\ T' \ | \ \lambda a. T$

    \textit{Kinds} $K ::= \ * \ |\ K \to K'$
    
    \textit{Type Environment} $\Gamma ::=  \cdot \ | \ \Gamma, a : T \ | \ \Gamma, c : T$

    \textit{Type Equivalence}  $(\lambda a.T)\ T' = [T'/a]T$

    \begin{tabular}{lllll}
\\
\infer{\Gamma \vdash (x|c) : T}{(x|c) : T \in \Gamma}    
&

&

\infer{\Gamma \vdash p_1\ p_2 : T}{\Gamma \vdash p_1 : T' \to T & \Gamma \vdash p_2 : T'}
&

&

\infer{\Gamma \vdash \lambda x : T'. p : T' \to T}{\Gamma, x : T' \vdash p : T}

\\ \\
\infer{\Gamma \vdash \lambda a . p: \forall a . T}{\Gamma \vdash p :  T & a \notin \mathrm{FV}(\Gamma)}
&

&

\infer{\Gamma \vdash p\ T' : [T'/a]T}{\Gamma \vdash p : \forall a . T}

&
&
\infer{\Gamma \vdash p : T'}{\Gamma \vdash p : T & T = T'}

  \end{tabular}  
    }
  \caption{System $\mathbf{F}_{\omega}$}
  \label{fig:systemfo}  
\end{figure}

We recall the standard System $\mathbf{F}_{\omega}$ in Figure \ref{fig:systemfo}. We use $c$ to denote term constant, $C$ to denote the type constants and $a, b$ to denote the type variables. We use $\mathrm{FV}(\Gamma)$ to mean all the free type variables in the environment $\Gamma$.
Since System $\mathbf{F}_{\omega}$ enjoys type level termination, we only need to work with normal form of a type. Note that the kind inference for $\mathbf{F}_{\omega}$ is decidable and we only work with well-kinded types\footnote{The kinding rules for $\mathbf{F}_{\omega}$ is available in Appendix \ref{krule}}. 

\subsection{Bidirectional second-order matching}
\label{bidirectional}
For the purpose of type checking and unification, we make the distinction between \textit{eigenvariables} and \textit{free variables} for the type variables. The type variable that can be substituted during the type checking or unification process is called \textit{free variables}, the variable
that cannot be substituted is called \textit{eigenvariables}.

We use $\mathrm{FV}(T)$ to denote the set of free variables in $T$ and $\mathrm{EV}(T)$ to denote
the set of eigenvariables in $T$. We use $\#$ as a predicate to denote the apartness of two sets and $\mathrm{agree}(\sigma)$
means that if $[T/a] \in \sigma$ and $[T'/a] \in \sigma$, then $T \equiv T'$. We
write $\mathrm{dom}(\sigma)\# \mathrm{codom}(\sigma)$ to means the free variables in the codomain
of $\sigma$ is disjoint with its domain. We say a type variable is first-order if it is of kind $*$. 


\begin{definition}[Dowek's bidirectional second-order matching]
  \label{bi-match}
Let $\theta_n^m(C)$ denote $\lambda a_1....\lambda a_n. C\ (b_1\ a_1 ...\ a_n)\ ...\ (b_m\ a_1 ...\ a_n)$, where $b_1,...,b_m$ are fresh free variables. Let $\pi_n^i$ denote the $i$-th projection $\lambda a_1 .... \lambda a_n . a_i$.  Let $V$ denotes a set of eigenvariables
  and $E$ denotes a set of unification problem $\{T_1 = T_1',. .., T_n = T_n'\}$. 

  We formalize the bidirectional second-order matching as a transition system from $(E, V, \sigma)$ to $(E', V', \sigma')$ in Figure \ref{fig:bi-match}.
 If $(\{T = T'\}, \emptyset, \mathrm{id}) \longrightarrow^* (\emptyset, V, \sigma)$, where $V \# \mathrm{EV}(\mathrm{codom}(\sigma)), \mathrm{dom}(\sigma)\# \mathrm{codom}(\sigma)$ and $\mathrm{agree}(\sigma)$,
  then we say the bidirectional matching is successful (denoted by $T\sim_\sigma T'$), otherwise it fails\footnote{Here $\mathrm{id}$ stands for identity substitution.}. 
  \begin{figure}

    \begin{tabular}{l}
    $(\{a = a, E\}, V, \sigma) \longrightarrow (\{E\}, V, \sigma)$.

    \\
    \\
    
    $(\{a = T, E\}, V, \sigma) \longrightarrow (\{[T/a]E\}, V, [T/a] \cdot \sigma)$,
    if $a$ is first-order, $a \notin \mathrm{FV}(T)$
    and $T \not \equiv a$.
    
    \\
    \\

    $(\{T = a, E\}, V, \sigma) \longrightarrow (\{a = T, E\}, V, \sigma)$. 

    \\
    \\

    $(\{\forall a . T = \forall b.T', E\}, V, \sigma) \longrightarrow_{\mathrm{forall}} (\{[a'/a]T = [a'/b]T', E\}, V\cup \{a'\}, \sigma)$,
    \\
    where $a'$ is a fresh eigenvariable.

    \\
    \\

    $(\{C\ T_1 \ ... \ T_n = C\ T_1'\ ... \ T_n', E\}, V, \sigma) \longrightarrow$
    $(\{T_1 = T_1',..., T_n = T_n', E\}, V, \sigma)$. 

    \\
    \\

    $(\{a\ T_1 \ ... \ T_n = C\ T_1'\ ... \ T_m', E\}, V, \sigma) \longrightarrow_{\mathrm{proj}}$
    $(\{T_i = C\ T_1'\ ... \ T_m', E \}, V, [\pi_n^i/a] \cdot \sigma)$.

    \\
    \\

    $(\{a\ T_1 \ ... \ T_n = C\ T_1'\ ... \ T_m', E\}, V, \sigma) \longrightarrow_{\mathrm{imi}}$
    \\
    $(\{(b_1 \ T_1 \ ... \ T_n)  = T_1', ..., (b_m \ T_1 \ ... \ T_n) = T_m', E, \}, V, [\theta_n^m(C)/a]\cdot \sigma)$,
    \\
    where $b_1, ..., b_m \in \mathrm{FV}(\theta_n^m(C))$.

    \\
    \\





    $(\{C\ T_1'\ ... \ T_m' = a\ T_1 \ ... \ T_n, E\}, V, \sigma) \longrightarrow_{\mathrm{exchange}} (\{a\ T_1 \ ... \ T_n = C\ T_1'\ ... \ T_m', E\}, V, \sigma)$.
    \end{tabular}
    \caption{Bidirectional second-order matching}
        \label{fig:bi-match}
\end{figure}
\end{definition}
The bidirectional matching algorithm in Figure \ref{fig:bi-match} is similar to the standard second-order matching, but it is bidirectional due to the exchange rule $\longrightarrow_{\mathrm{exchange}}$. Moreover, unlike standard second-order matching, in the rules $\longrightarrow_{\mathrm{proj}}$ and $\longrightarrow_{\mathrm{imi}}$, we do not perform substitution on $E$, this ensures the termination of the transition system. 
We also add the $\longrightarrow_{\mathrm{forall}}$ rule to handle the forall-quantified types.
The bidrectional second-order matching is sound and terminating. The rules $\longrightarrow_{\mathrm{proj}}$ and $\longrightarrow_{\mathrm{imi}}$ are overlapped, so there can be
multiple unifiers for a given unification problem. 


\begin{theorem}[Termination and Soundness\footnote{The proof is at Appendix \ref{bisound}}]
  \label{unif:sound}
  The transition system in Figure \ref{fig:bi-match} is terminating. Moreover, 
  if $T\sim_\sigma T'$, then $\sigma T \equiv \sigma T'$.
\end{theorem}

\begin{example}
  Consider the unification problem of $\forall a . a \to a$ and $\forall a . q \to a$.
  They should not be unified. Indeed it will not be a successful becase we will have the following transition:

  \begin{center}
    $(\{\forall a . a \to a = \forall a . q \to a\}, \emptyset, \mathrm{id}) \longrightarrow_{\mathrm{forall}} (\{ a_1 \to a_1 =  q \to a_1\}, \{a_1\}, \mathrm{id}) \longrightarrow (\{  a_1 =  q, a_1 = a_1\}, \{a_1\}, \mathrm{id}) \longrightarrow^* (\emptyset, \{a_1\}, [a_1/q])$
  \end{center}
\noindent  But $[a_1/q]$ is not a unifier because its codomain is not apart from $\{a_1\}$. 
\end{example}

\begin{example}
  Consider the unification problem of $p\ \mathsf{Bot} \to p\ \mathsf{Top}$ and $\mathsf{Bot} \to\mathsf{Top}$.
  The first step of the transition is:
  $(\{p\ \mathsf{Bot} \to p\ \mathsf{Top} = \mathsf{Bot} \to \mathsf{Top}\}, \emptyset, \mathrm{id}) \longrightarrow (\{p\ \mathsf{Bot} = \mathsf{Bot}, p\ \mathsf{Top} = \mathsf{Top} \}, \emptyset, \mathrm{id})$. Then there will be the following four possible transitions, but only the last one is successful. 

  \noindent 1. {\footnotesize $(\{p\ \mathsf{Bot} = \mathsf{Bot}, p\ \mathsf{Top} = \mathsf{Top} \}, \emptyset, \mathrm{id}) \longrightarrow_{\mathrm{imi}} (\{\mathsf{Bot} = \mathsf{Bot}, p\ \mathsf{Top} = \mathsf{Top} \}, \emptyset, [\lambda x. \mathsf{Bot}/p]) $

    $\longrightarrow (\{p\ \mathsf{Top} = \mathsf{Top} \}, \emptyset, [\lambda x. \mathsf{Bot}/p]) \longrightarrow_{\mathrm{imi}} (\emptyset, \emptyset, [\lambda x. \mathsf{Top}/p,\lambda x. \mathsf{Bot}/p])$}

\

\noindent 2. {\footnotesize $(\{p\ \mathsf{Bot} = \mathsf{Bot}, p\ \mathsf{Top} = \mathsf{Top} \}, \emptyset, \mathrm{id}) \longrightarrow_{\mathrm{imi}} (\{\mathsf{Bot} = \mathsf{Bot}, p\ \mathsf{Top} = \mathsf{Top} \}, \emptyset, [\lambda x. \mathsf{Bot}/p]) $

  $\longrightarrow (\{p\ \mathsf{Top} = \mathsf{Top} \}, \emptyset, [\lambda x. \mathsf{Bot}/p]) \longrightarrow_{\mathrm{proj}} (\emptyset, \emptyset, [\lambda x. x/p,\lambda x. \mathsf{Bot}/p])$}

\

\noindent 3. {\footnotesize $(\{p\ \mathsf{Bot} = \mathsf{Bot}, p\ \mathsf{Top} = \mathsf{Top} \}, \emptyset, \mathrm{id}) \longrightarrow_{\mathrm{proj}} (\{\mathsf{Bot} = \mathsf{Bot}, p\ \mathsf{Top} = \mathsf{Top} \}, \emptyset, [\lambda x. x/p])$

  $\longrightarrow (\{p\ \mathsf{Top} = \mathsf{Top} \}, \emptyset, [\lambda x. \mathsf{Bot}/p]) \longrightarrow_{\mathrm{imi}} (\emptyset, \emptyset, [\lambda x. \mathsf{Top}/p,\lambda x. x/p])$}

\

\noindent 4. {\footnotesize $(\{p\ \mathsf{Bot} = \mathsf{Bot}, p\ \mathsf{Top} = \mathsf{Top} \}, \emptyset, \mathrm{id}) \longrightarrow_{\mathrm{proj}} (\{\mathsf{Bot} = \mathsf{Bot}, p\ \mathsf{Top} = \mathsf{Top} \}, \emptyset, [\lambda x. x/p])$

  $\longrightarrow (\{p\ \mathsf{Top} = \mathsf{Top} \}, \emptyset, [\lambda x. \mathsf{Bot}/p]) \longrightarrow_{\mathrm{proj}} (\emptyset, \emptyset, [\lambda x. x/p,\lambda x. x/p])$}

\end{example}

\subsection{The type checking algorithm}
\label{tcalg}
Let $L$ be a list of pairs $(a, n)$, where $a$ is a type variable (free variable or eigenvariable) and $0 \leq n$. We call such $n$ a \textit{scope value}. We write $L(a)$ to mean the scope value of $a$, when $a \notin L$, $L(a)$ is defined to be an arbitrary large value. For a set of variables $S$, we write $L(S)$ to mean the set of its scope values in $L$. We define $\mathrm{max}(L)$ to be the maximum scope value in $L$, if $L$ is empty, then we set $\mathrm{max}(L) = 0$.  
The following definition of \textit{scope check} ensures that the free variables can only
be substituted with the types that contains eigenvariables that are introduced before.  

\begin{definition}[Scope check] We define $\mathsf{Scope}(L, \sigma)$ to be the following predicates: For any $a\in \mathrm{dom}(\sigma)$, if $(a, n) \in L$, then for any $b \in \mathrm{EV}(\sigma a)$, we have $(b, n') \in L$ and $n' < n$. 
\end{definition}


Let $\mathrm{fv}(a, \sigma) = \mathrm{FV}(\mathrm{codom}(\sigma a))$ and $L_{\mathrm{eigen}} = [(a, n) | (a, n) \in L, \mathrm{isEigen}(a)]$. The following definition
of $\sigma L$ will replace the pair $(a, n) \in L$ (where $a \in \mathrm{dom}(\sigma)$),  
by the pairs $(b, n')$, where $b \in \mathrm{fv}(a, \sigma)$ and $n'$ is the minimal one among $n$ and $L(\mathrm{fv}(a, \sigma))$. We use $L + L'$ to mean append $L, L'$. We write
$|L|$ to means a scope environment that has the same variables as $L$, but if a variable
has multiple scope values in $L$, then it will have the minimal one in $|L|$. 

\begin{definition}[Updating]
  \
  
\noindent Let $M^a = \mathrm{min}\{L(a), L(\mathrm{fv}(a, \sigma))\}$ and $L' = [(b, M^a) | a \in \mathrm{dom}(\sigma), b \in \mathrm{fv}(a, \sigma)]$. We define $\sigma L = |L'| + L_{\mathrm{eigen}}$.  
\end{definition}

Let $\Psi ::= \cdot \ | \ (L, \Gamma, e, T), \Psi \ | \  (L, \Gamma, \diamond, \diamond), \Psi$. The tuple $(L, \Gamma, e, T)$ means $e$ is an unannotated program to be type checked with the type $T$ under the scope environment $L$ and the typing environment $\Gamma$. The tuple $(L, \Gamma, \diamond, \diamond)$ means the type checking process for a branch is finished,
with the final scope and typing environment $L$ and $\Gamma$. We now define the type checking
algorithm as a transition system between $(\Psi, \sigma)$. We write $\sigma \Gamma$ to mean
applying $\sigma$ to all the types in $\Gamma$. We use $(x|c)$ to mean a program variable $x$ or a data constructor $c$. Furthermore, $T_1,..., T_n \to T$ means $T_1 \to \ ...\ \to T_n \to T$.

\begin{definition}[A type checking algorithm]
  \label{rsm}
  \fbox{$(\Psi, \sigma) \longrightarrow (\Psi', \sigma')$}

  \begin{itemize}
  \item $(\{(L, \Gamma, (x | c), T), \Psi\}, \sigma) \longrightarrow_s (\{\sigma' \Psi, (\sigma'L', \sigma'\Gamma, \diamond, \diamond)\}, \sigma' \cdot \sigma)$

    if $(x|c) : \forall a_1 ... \forall a_k . T' \in \Gamma$ and $\mathsf{Scope}(L, \sigma')$. Here $\sigma', L'$ is defined by the followings.
    \begin{itemize}
    \item $T' \sim_{\sigma'} T$.

    \item $L' = L + [(a_i, n)\ |\ 0 \leq i \leq k, n = \mathrm{max}(L)+1]$. Note that $a_1,..., a_k$ are the fresh free variables in $T'$. 

    \end{itemize}

    \
    
  \item $(\{(L, \Gamma, \lambda x_1.... \lambda x_n . e\ e_1 ... \ e_k, \forall a_1. ...\forall a_m . T_1,..., T_n \to T), \Psi\}, \sigma) \longrightarrow_i$

    $(\{([L, (a_1, n'+1),..., (a_m, n'+m)], [\Gamma, x_1 : T_1,..., x_n : T_n], e, T), \Psi\}, \sigma)$,

    where $m > 0$, $k, n \geq 0 $, $n' = \mathrm{max}(L)$ and $a_1,..., a_m$ are the fresh eigenvariables in $T_1,..., T_n \to T$. 

        \
        
  \item  $(\{(L, \Gamma, (x | c)\ e_1 \ ... \ e_n, T), \Psi \}, \sigma) \longrightarrow_a$ 

    $(\{(\sigma' L', \sigma'\Gamma, e_1, \sigma' T_1'), ..., (\sigma' L', \sigma'\Gamma, e_l, \sigma' T_l'), $

    \quad \quad $(\sigma' L', \sigma'\Gamma, e_{l+1}, \sigma' b_{1}), ..., (\sigma' L', \sigma'\Gamma, e_{n}, \sigma' b_{n-l}), \sigma' \Psi\}, \sigma' \cdot \sigma)$,

    where $0 \leq l \leq n$, $n > 0$, $(x | c) : \forall a_1 ... \forall a_k. T_1',..., T_l' \to T' \in \Gamma$ and $\mathsf{Scope}(L, \sigma')$. Here $\sigma', L'$ is defined by the following.

    \begin{itemize}
    \item $T' \sim_{\sigma'} (b_{1} ,..., b_{n-l} \to T)$, where $b_{1},..., b_{n-l}$ are fresh free variables.
    \item $n'= \mathrm{max}(L)+1$, $L' = L + [(a_i, n') | 1 \leq i \leq k] + [(b_j, n') | 1\leq b_j \leq n-l]$. Note that $a_1,..., a_k$ are the fresh free variables in $T_1',..., T_l' \to T'$. 
    \end{itemize}

    \
    
  \item $(\{(L, \Gamma, (x | c)\ e_1 \ ... \ e_n, T), \Psi\}, \sigma) \longrightarrow_b$
    
    $(\{(\sigma' L', \sigma'\Gamma,e_1, \sigma' T_1'), ..., (\sigma' L', \sigma'\Gamma, e_n, \sigma' T_n'), \sigma' \Psi\}, \sigma' \cdot \sigma)$,

    where $0< n < l$, $(x | c) : \forall a_1 ... \forall a_k. T_1',..., T_l' \to T' \in \Gamma$
    and $\mathsf{Scope}(L, \sigma')$. Here $\sigma', L'$ is defined by the following.

    \begin{itemize}
    \item $(T_{n+1}',..., T_l' \to T') \sim_{\sigma'} T$
    \item $L' = L + [(a_i, n') \ | 1 \leq i \leq k, n'= \mathrm{max}(L)+1 ]$. Note that $a_1,..., a_k$ are the fresh free variables in $T_1',..., T_l' \to T'$. 
    \end{itemize}
  \end{itemize}

  \noindent  Note that $\sigma \Psi$ is defined as followings.
  
  $\sigma \{\} = \{\}$

  $\sigma \{(L, \Gamma, e, T), \Psi\} = \{(\sigma L, \sigma \Gamma, e, \sigma T), \sigma \Psi\}$, where $\mathsf{Scope}(L, \sigma)$. 

  $\sigma \{(L, \Gamma, \diamond, \diamond), \Psi\} = \{(\sigma L, \sigma \Gamma, \diamond, \diamond), \sigma \Psi\}$, where $\mathsf{Scope}(L, \sigma)$.

\end{definition}

The transition system is defined
over the structure of $e$ and $T$, hence it is terminating. To type
check $e$ with $T$ under the environment $\Gamma$, the initial
state will be $(\{([], \Gamma, e, T)\}, \mathrm{id})$.
We say the type checking is successful if the final state is of the form
$(\{(L_1, \Gamma_1, \diamond, \diamond),..., (L_n, \Gamma_n,\diamond,\diamond)\},\sigma)$, where $n > 0$.

There are two kinds of nondeterminism going on in the transition system in Definition \ref{rsm}. One is due to our decision on handling type abstraction, i.e. the transition
will be branching when
we check an application against a forall-quantified type. This means the rule $\longrightarrow_i$ is overlapped with rules $\longrightarrow_s$ when $n = k = 0$, and it is overlapped with $\longrightarrow_a$ and $\longrightarrow_b$ when $n = 0, k > 0$. Another kind of nondeterminism
is due to the appearance of the second-order type variables, so the rules $\longrightarrow_s, \longrightarrow_a, \longrightarrow_b$ could leads to multiple states, as the bidirectional matching can generate multiple valid unifiers.

Each of the transitions $\longrightarrow_s, \longrightarrow_a, \longrightarrow_b $
will generate a substitution $\sigma$, which will be checked by the predicate
$\mathsf{Scope}$ against the current scope value environment $L$. Then
 $L$ is extended to $L'$ with some new free variables, 
and this $L'$ will be updated to a new scope environment $\sigma L'$,
which will contain all the free variables that appear in the environment $\sigma \Gamma$.

The $\longrightarrow_s$ rule is for handling the variable and the constant case. The $\longrightarrow_i$ rule is solely responsible for the type abstraction of the outermost forall-quantified variables. When checking a lambda-abstraction against a forall-type, it is only natural to perform type abstraction. This is why $\longrightarrow_i$ rule is also removing the lambda-abstractions after the type abstraction.  


In general, a function of a type $T_1,..., T_n \to T$ does not always have $n$ input.
The rule $\longrightarrow_b$ accounts for the partial application, while The rule $\longrightarrow_a$
accounts for the \textit{over application}. For example, it seems $\mathsf{id} : \forall a . a \to a$
can only take one input, but we know $\mathsf{id}\ \mathsf{id}\ \mathsf{id}$ is typable with $\forall a . a \to a$, this is why we need the rule $\longrightarrow_a$ to provide additional free type variables for later instantiation (i.e. the free variables $b_1,..., b_{n-l}$ in the rule $\longrightarrow_a$).

\subsection{Soundness and examples}

We prove the type checking algorithm is sound. The proof gives rise to an algorithm
that produces an annotated program if the type checking is successful. 

\begin{theorem}[Soundness\footnote{The proof is at Appendix \ref{sound:type}}]
  \
  
\noindent If $(\{([] , \Gamma, e, T)\}, \mathrm{id}) \longrightarrow^* (\{(L_1, \Gamma_1, \diamond, \diamond),..., (L_n, \Gamma_n,\diamond,\diamond)\}, \sigma)$, where $\mathrm{FV}(\Gamma) = \emptyset$ and $\mathrm{FV}(T) = \emptyset$, then there exists a $p$ in $\mathbf{F}_\omega$ such that $ \Gamma \vdash \sigma p : T$ and $|\sigma p| = e$\footnote{Here $|p|$ means erasing all the type annotations in $p$.}. 
\end{theorem}

\begin{example}
  \label{church}
 Consider the following Church encoded numbers.
\begin{verbatim}
type Nat :: * = forall x . (x -> x) -> x -> x
zero :: Nat
zero = \ s z -> z
succ :: Nat -> Nat
succ n = \ s z -> s (n s z)
add :: Nat -> Nat -> Nat
add n m = n succ m
\end{verbatim}
To type check \texttt{add}, let $\Gamma = \mathsf{zero} : \mathsf{Nat}, \mathsf{succ} : \mathsf{Nat} \to \mathsf{Nat}$ and the initial state be $(\{([], \Gamma, \lambda n . \lambda m . n \ \mathsf{succ}\ m, \mathsf{Nat} \to \mathsf{Nat} \to \mathsf{Nat})\}, \mathrm{id})$. We have a successful and a failed transition in Figure \ref{fig:trans}. Note that $\mathsf{Nat}$
is an abbreviation of $\forall x . (x \to x) \to x \to x$. A branching occurs
at the state $([], [\Gamma, n : \mathsf{Nat}, m : \mathsf{Nat}], n \ \mathsf{succ}\ m, \mathsf{Nat})$, where we can apply either $\longrightarrow_a$ or $\longrightarrow_i$, the former will lead to
a successful transition, while the latter will fail because $\mathsf{Nat} \to \mathsf{Nat}$
cannot be unified with $((x_0 \to x_0) \to x_0 \to x_0) \to ((x_0 \to x_0) \to x_0 \to x_0)$. 


\begin{figure}
1. $(\{([], \Gamma, \lambda n . \lambda m . n \ \mathsf{succ}\ m, \mathsf{Nat} \to \mathsf{Nat} \to \mathsf{Nat})\}, \mathrm{id}) \longrightarrow_i $

 $(\{([], [\Gamma, n : \mathsf{Nat}, m : \mathsf{Nat}], n \ \mathsf{succ}\ m, \mathsf{Nat})\}, \mathrm{id}) = $

 $(\{([], [\Gamma, n : \forall x . (x \to x) \to x \to x, m : \mathsf{Nat}], n \ \mathsf{succ}\ m, \mathsf{Nat})\}, \mathrm{id}) \longrightarrow_a$

  $(\{([], [\Gamma, n : \mathsf{Nat}, m : \mathsf{Nat}], \mathsf{succ}, \mathsf{Nat} \to \mathsf{Nat}), ([], [\Gamma, n : \mathsf{Nat}, m : \mathsf{Nat}], m, \mathsf{Nat})\}, [\mathsf{Nat}/x])$

  $ \longrightarrow_s^* (\{([], [\Gamma, n : \mathsf{Nat}, m : \mathsf{Nat}], \diamond, \diamond), ([], [\Gamma, n : \mathsf{Nat}, m : \mathsf{Nat}], \diamond, \diamond)\}, [\mathsf{Nat}/x])$






  \

  2.  $(\{([], \Gamma, \lambda n . \lambda m . n \ \mathsf{succ}\ m, \mathsf{Nat} \to \mathsf{Nat} \to \mathsf{Nat})\}, \mathrm{id}) \longrightarrow_i $

 $(\{([], [\Gamma, n : \mathsf{Nat}, m : \mathsf{Nat}], n \ \mathsf{succ}\ m, \forall x . (x \to x) \to x \to x)\}, \mathrm{id}) \longrightarrow_i $

  $(\{([(x_0, 1)], [\Gamma, n : \mathsf{Nat}, m : \mathsf{Nat}], n \ \mathsf{succ}\ m, (x_0 \to x_0) \to x_0 \to x_0)\}, \mathrm{id}) \longrightarrow_a$
  
  $(\{([(x_0, 1)], [\Gamma, n : \mathsf{Nat}, m : \mathsf{Nat}], \mathsf{succ}, ((x_0 \to x_0) \to x_0 \to x_0) \to ((x_0 \to x_0) \to x_0 \to x_0)), ([(x_0, 1)], [\Gamma, n : \mathsf{Nat}, m : \mathsf{Nat}], m, (x_0 \to x_0) \to x_0 \to x_0)\}, \mathrm{id}) \not \longrightarrow$


  \caption{The type checking transition of Example \ref{church}}
\label{fig:trans}  
\end{figure}
\end{example}

\begin{example}
Let $\Gamma = \mathsf{k1} : \forall q . (\forall y . q \to y) \to \mathsf{Bot}, \mathsf{k2} : \forall x . x \to x$. To type check $\mathsf{k1}\ \mathsf{k2}$ with type $\mathsf{Bot}$, let the initial state be $(\{([], \Gamma, \mathsf{k1}\ \mathsf{k2}, \mathsf{Bot})\}, \mathrm{id})$. We will have the following two unsuccessful transitions:

1. $(\{([], \Gamma,  \ \mathsf{k1}\ \mathsf{k2}, \mathsf{Bot})\}, \mathrm{id}) \longrightarrow_a
  (\{([(q, 1)], \Gamma, \mathsf{k2}, \forall y. q \to y)\}, \mathrm{id}) \not \longrightarrow_s$

\

2. $(\{([], \Gamma,  \ \mathsf{k1}\ \mathsf{k2}, \mathsf{Bot})\}, \mathrm{id}) \longrightarrow_a
  (\{([(q, 1)], \Gamma, \mathsf{k2}, \forall y. q \to y)\}, \mathrm{id}) \longrightarrow_i $

$(\{([(q, 1), (y_0, 2)], \Gamma, \mathsf{k2}, q \to y_0)\}, \mathrm{id}) \not \longrightarrow_s$

\noindent In the first transition, the step $\longrightarrow_s$ can not be performed because
$x\to x$ is not unifiable with $\forall y . q \to y$. In the second transition, the step $\longrightarrow_s$ can not be performed because the unifier of $x \to x$ and $q \to y_0$ is $[y_0/q,y_0/x]$, but the predicate $\mathsf{Scope}([(q, 1), (y_0, 2)], [y_0/q,y_0/x])$ is false because $q$ refers to the eigenvariable $y_0$, which is introduced later than $q$. 
\end{example}

\subsection{Discussion}

The type checking algorithm cannot type check beta-redex, i.e. any programs of the form
$(\lambda x.e)\ e'$, because it cannot infer a type for $\lambda x.e$. So programmer will have to
restructure the program as $f_1 = \lambda x.e, f_2 = f_1 \ e'$ and supply type annotations
for $f_1, f_2$. We argue that this is not a serious problem because most
programs do not contain explicit beta-redex (See also the programs examples in the Appendix \ref{examples}).

Although the order of the tuples in $\Psi$ does not matter for the soundness proof, it does matter in practice. The transitions $\longrightarrow_a, \longrightarrow_b$ could generate multiple
new tuples, how are we going to decide in what order to check these tuples? We could
try all the possible combinations, but this is not very efficient. So in the prototype implementation
we use a measure to arrange the order of tuples generated by $\longrightarrow_a, \longrightarrow_b$. The measure is the number of implications in the goal $T$ in
a tuple $(L, \Gamma, e, T)$, the more implications it has, the higher priority we will
give to check this tuple (as it may provide more useful information that we can use later).
For example \cite{jones2007practical}, let $\Gamma = \mathsf{revapp} : \forall a. \forall b . a \to (a \to b) \to b, \mathsf{poly} : (\forall v. v \to v) \to \mathsf{Pair}\ \mathsf{Nat}\ \mathsf{Bool}$. Here $\mathsf{Pair}$ is a type constructor. Consider the following transition.

$(\{([], \Gamma, \mathsf{revapp}\ (\lambda x.x)\ \mathsf{poly}, \mathsf{Pair}\ \mathsf{Nat}\ \mathsf{Bool})\},\mathrm{id}) \longrightarrow_a $

$(\{([(a_0, 1)], \Gamma, (\lambda x.x), a_0), $

\quad $([(a_0, 1)], \Gamma, \mathsf{poly}, a_0 \to \mathsf{Pair}\ \mathsf{Nat}\ \mathsf{Bool})\},[\mathsf{Pair}\ \mathsf{Nat}\ \mathsf{Bool}/b])$

\noindent Here $a_0$ is a free variable that is introduced when we instantiate the type of $\mathsf{revapp}$. If we
try to type check the tuple $([(a_0, 0)], \Gamma, (\lambda x.x), a_0)$ first, we will stuck
because no rule apply to this tuple. But if we type check the tuple $([(a_0, 0)], \Gamma, \mathsf{poly}, a_0 \to \mathsf{Pair}\ \mathsf{Nat}\ \mathsf{Bool})$ first, we will obtain the new information, i.e. $a_0$ will be instantiated with $(\forall v. v\to v)$, as a result, we can type check the tuple $([], \Gamma, (\lambda x.x), \forall v. v\to v)$. This example fits the heuristic
that the type $a_0 \to \mathsf{Pair}\ \mathsf{Nat}\ \mathsf{Bool}$ gives more information than the type $a_0$
since it has more implications. 


\section{Extensions and implementation}
\label{extent}

In order to show the type checking algorithm works for a nontrivial
subclass of functional programs, we extend it to work with
let-bindings and pattern matching. First we extend the unannotated expression, $e\ ::= ... \ | \ \mathrm{let}\ x = e\ \mathrm{in}\ e' \ | \ \mathrm{let}\ x : T = e\ \mathrm{in}\ e'\ | \ \mathrm{case}\ e\ \{p_i \to e_i\}_{i\in N}$. Here $N$ stands for an index set and $p_i$ stands for the pattern $p\ ::=\ x \ | \ c \ | \ c \ p_1\ ... \ p_n$. The following are
the rules for checking let-bindings and pattern matching. The idea is that 
we can use fresh free variables as goals to enable the algorithm to perform a limited degree of inference. 

\begin{definition}[Extensions]
  \label{extensions}
  \
  
  \begin{itemize}

    
  \item $(\{(L, \Gamma, \mathrm{let}\ x = e\ \mathrm{in}\ e', T), \Psi\}, \sigma) \longrightarrow_{\mathrm{let}}$

    $(\{([L, (b, n)], [\Gamma, x : b], e, b), ([L, (b, n)], [\Gamma, x : b], e', T), \Psi\}, \sigma)$, where $b$ is a fresh free variable and $n = \mathrm{max}(L)+1$.

    \
    
      \item $(\{(L, \Gamma, \mathrm{let}\ x : T = e\ \mathrm{in}\ e', T'), \Psi\}, \sigma) \longrightarrow_{\mathrm{let'}}$

    $(\{(L, [\Gamma, x : T], e, T), (L, [\Gamma, x : T], e', T'), \Psi\}, \sigma)$. 

    \
    
  \item $(\{(L, \Gamma, \mathrm{case}\ e\ \{p_i \to e_i\}_{i\in N}, T), \Psi\}, \sigma) \longrightarrow_{\mathrm{case}}$

    $(\{([L, (b, n)], \Gamma, e, b), $

    \quad $\{([L, (b, n), \mathrm{SC}(\Phi_{p_i})], [\Gamma, \Phi_{p_i}], p_i, b),$

    \quad \quad $([L, (b, n), \mathrm{SC}(\Phi_{p_i})], [\Gamma, \Phi_{p_i}], e_i, T)\}_{i\in N}, \Psi\}, \sigma)$.

    Here $b$ is fresh free variable, $n = \mathrm{max}(L)+1$, $\Phi_{p_i} = [(x: a) \ | \ x \in \mathrm{FV}(p_i), \mathrm{freshFree}(a)]$ and $\mathrm{SC}(\Phi_{p_i}) = [(a, n)\ |\ a \in \mathrm{codom}(\Phi_{p_i})]$.  
  \end{itemize}
\end{definition}

The type checking rule $\longrightarrow_{\mathrm{let}}$ views the let-bind variable as an abbreviation, hence it does not support let-generalization.  
To support the let-generalization, we use the rule $\longrightarrow_{\mathrm{let'}}$, which requires the user to supply annotation for the let-bind variable. This view of let-bindings coincides with the one in \cite{vytiniotis2010let}. 

We have implemented the Definition \ref{rsm}, \ref{extensions} in a prototype type checker. Benefits from the soundness proof, the type checker can
output an annotated program. Thus we can use a separated proof checker to perform an additional check on the annotated program. 
We use the type checker to type check Stump's Church-encoded merge sort algorithm that uses impredicative Church-encoded list and braun tree \cite{stump2016efficiency}. To show the support of second-order types, we use the type checker to type check Bird and Paterson's program \cite{bird1999bruijn} without the dupplications of the generalized fold\footnote{See Appendix \ref{examples} for more details.}. 

\section{Conclusion}

Higher-rank types have been well-studied in the literature (e.g. \cite{odersky1996putting}, \cite{vytiniotis2006boxy}, \cite{jones2007practical}, \cite{dunfield2013complete}). 
There are also a lot of research on impredicative types (e.g. \cite{jones1997first}, \cite{vytiniotis2006boxy}, \cite{le2014mlf}). The study on second-order types is relatively
few, but it has been considered before (e.g. \cite{pfenning1988partial}, \cite{neubauer2002type}). The type checking algorithm in this paper differs from most existing ones in two aspects. (1) It exploits nondeterminism to handle type abstraction and second-order types while being terminating. (2) It is not an extension of Hindley-Milner algorithm and cannot perform general type inference.
Hence the algorithm represents
a particular approach to type check higher-rank, impredicative and second-order types. 

To summarize, we propose a type checking algorithm that supports 
 higher-rank, impredicative and second-order types. The algorithm relies on a
specialized version of second-order unification, which we studied. We also
prove the type checking algorithm is sound. The potential benefits includes
 a cleaner support for impredicative Church-encoded programs and an elimination of code duplications due to the
use of second-order and impredicative types. 
\bibliographystyle{plain}
  
\bibliography{higher-rank}
\newpage

\appendix

\section{Kinding rules}
\label{krule}
\begin{definition}[Erasure]
  \label{erase}
  \
  
 $|c| = c$ \quad  $|x| =x $ \quad $|\lambda x : T . p| = \lambda x . |p|$
  \quad $|p\ p'| = |p|\ |p'| $ \quad   $|\lambda a.p| = |p| $ \quad $|p\ T| = |p|$
\end{definition}

We define $\Delta ::=  \cdot \ | \ \Delta, a : K \ | \ \Delta, C : K$. 

\begin{definition}[Kinding rules]
  \label{kindsystem}
\fbox{$\Delta \vdash T : K$}
  
\

{
\begin{tabular}{lll}
\\
\infer{\Delta \vdash a|C : K}{(a|C : K) \in \Delta}    
&

\infer[]{\Delta \vdash T_1\ T_2 : K}{\Delta \vdash T_1 : K' \to K & \Delta \vdash T_2 : K'}
&

\infer[]{\Delta \vdash \lambda a . T : K' \to K}{\Delta, a : K' \vdash T : K}
\\ \\
\infer{\Delta \vdash T \to T': *}{\Delta \vdash T :  * & \Delta \vdash T' :  *}
&

\infer{\Delta \vdash \forall a. T : *}{\Delta \vdash T : *}

&

  \end{tabular}  
}
\end{definition}

\section{Dowek's bidrectional second-order matching}
\label{bisound}
Recall Dowek's bidirectional second-order matching.

\begin{definition}[Dowek's bidirectional second-order matching]
  \
  
  \begin{tabular}{l}

    \\
    $(\{a = a, E\}, V, \sigma) \longrightarrow (\{E\}, V, \sigma)$.

    \\
    \\
    
    $(\{a = T, E\}, V, \sigma) \longrightarrow (\{[T/a]E\}, V, [T/a] \cdot \sigma)$,
    \\
    if $a$ is first-order, $a \notin \mathrm{FV}(T)$ and $T \not \equiv a$.
    
    \\
    \\

    $(\{T = a, E\}, V, \sigma) \longrightarrow (\{a = T, E\}, V, \sigma)$. 

    \\
    \\

    $(\{\forall a . T = \forall b.T', E\}, V, \sigma) \longrightarrow (\{[a'/a]T = [a'/b]T', E\}, V\cup \{a'\}, \sigma)$,
    \\
    where $a'$ is a fresh eigenvariable.

    \\
    \\

    $(\{C\ T_1 \ ... \ T_n = C\ T_1'\ ... \ T_n', E\}, V, \sigma) \longrightarrow$
    $(\{T_1 = T_1',..., T_n = T_n', E\}, V, \sigma)$. 

    \\
    \\

    $(\{a\ T_1 \ ... \ T_n = C\ T_1'\ ... \ T_m', E\}, V, \sigma) \longrightarrow$
    $(\{T_i = C\ T_1'\ ... \ T_m', E \}, V, [\pi_n^i/a] \cdot \sigma)$.

    \\
    \\

    $(\{a\ T_1 \ ... \ T_n = C\ T_1'\ ... \ T_m', E\}, V, \sigma) \longrightarrow$
    \\
    $(\{(b_1 \ T_1 \ ... \ T_n)  = T_1', ..., (b_m \ T_1 \ ... \ T_n) = T_m', E, \}, V, [\theta_n^m(C)/a]\cdot \sigma)$,
    \\
    where $b_1, ..., b_m \in \mathrm{FV}(\theta_n^m(C))$.

    \\
    \\





    $(\{C\ T_1'\ ... \ T_m' = a\ T_1 \ ... \ T_n, E\}, V, \sigma) \longrightarrow (\{a\ T_1 \ ... \ T_n = C\ T_1'\ ... \ T_m', E\}, V, \sigma)$.
  \end{tabular}

  \
  
  If $(\{T = T'\}, \emptyset, \mathrm{id}) \longrightarrow^* (\emptyset, V, \sigma)$, $V \# \mathrm{EV}(\mathrm{codom}(\sigma)), \mathrm{dom}(\sigma)\# \mathrm{codom}(\sigma)$ and $\mathrm{agree}(\sigma)$,
  then we say the bidirectional matching is successful, otherwise it fails. 

\end{definition}

We use $\forall \underline{b}. T$ to denote $\forall b_1 .... \forall b_n . T$ for some $n \geq 0$. We call equations $\forall \underline{b} . C\ T_1\ ...\ T_n = \forall \underline{c} . C\ T_1'\ ...\ T_n'$, $\forall \underline{b} . b_i\ T_1\ ...\ T_n = \forall \underline{c} . C\ T_1'\ ...\ T_n'$ and $\forall \underline{b} . b_i\ T_1\ ...\ T_n = \forall \underline{c} . c_i\ T_1'\ ...\ T_n'$  \textit{rigid-rigid} equations
and $\forall \underline{b}. a \ T_1 \ ... T_n \ = \forall \underline{c}. C\ T_1'\ ...\ T_m'$(or $\forall \underline{c}. C\ T_1'\ ...\ T_m' = \forall \underline{b}. a \ T_1 \ ...\ T_n$, where $a \notin \{b_1,.., b_n\}$) \textit{rigid-flexible} equations. For the rigid-flexible equations such as $\forall \underline{b}. a \ T_1 \ ...\ T_n \ = \forall \underline{c}. C\ T_1'\ ...\ T_m'$, we call $C\ T_1'\ ...\ T_m'$ a \textit{rigid component}. 

\begin{theorem}
  $\longrightarrow$ is a terminating transition system. 
\end{theorem}
\begin{proof}
  Let $n_{\mathrm{fvar}}$ denote the number of \textit{distinct} first-order variables, $n_{\mathrm{rr}}$
  the total size of rigid-rigid equations, $n_{\mathrm{rcomp}}$
  the sum of the sizes of the rigid components, $n_{\mathrm{svar}}$ the number of occurrences of all
  the second-order variables, $n_\forall$ denotes the number of equations of the form $\forall a . T = \forall b . T'$, $n_{\mathrm{eq}}$ the number of equations of the form $T = a$ or
  $C \ T_1 \ ... \ T_n = a\ T_1'\ ...\ T_m' $.
  Consider the lexicographic ordering on $(n_{\mathrm{fvar}}, n_{\mathrm{rcomp}}, n_{\mathrm{svar}}, n_{\mathrm{rr}}, n_\forall, n_{\mathrm{eq}})$. We will show this ordering is decreasing strictly along the transition. 

  \begin{itemize}
  \item Case: $(\{a = a, E\}, V, \sigma) \longrightarrow (\{E\}, V, \sigma)$.

    \
    
    In this case $n_{\mathrm{fvar}}, n_{\mathrm{svar}}$ is either not changed or become strictly
    smaller,
    $n_{\mathrm{rcomp}}, n_{\mathrm{rr}}, n_\forall$ does not change, $n_{\mathrm{eq}}$ becomes strictly smaller.

  \item Case: $(\{a = T, E\}, V, \sigma) \longrightarrow (\{[T/a]E\}, V, [T/a] \cdot \sigma)$, if $a$ is first-order, $a \notin \mathrm{FV}(T)$
    and $T \not \equiv a$.

    \

    In this case $n_{\mathrm{fvar}}$ becomes strictly smaller.

   \item Case: $(\{T = a, E\}, V, \sigma) \longrightarrow (\{a = T, E\}, V, \sigma)$. 

     \

     In this case only $n_{\mathrm{eq}}$ becomes strictly smaller.

   \item Case:

     $(\{C\ T_1'\ ... \ T_m' = a\ T_1 \ ... \ T_n, E\}, V, \sigma) \longrightarrow (\{a\ T_1 \ ... \ T_n = C\ T_1'\ ... \ T_m', E\}, V, \sigma)$.

      \

      In this case only $n_{\mathrm{eq}}$ becomes strictly smaller.

    \item Case. $(\{\forall a . T = \forall b.T', E\}, V, \sigma) \longrightarrow (\{[a'/a]T = [a'/b]T', E\}, V\cup \{a'\}, \sigma)$,
    where $a'$ is a fresh eigenvariable.

    \
    
    In this case $n_{\mathrm{fvar}}, n_{\mathrm{rcomp}}, n_{\mathrm{svar}}, n_{\mathrm{rr}}$ do not change, $n_\forall$ becomes strictly smaller.

    \item Case. 

    $(\{C\ T_1 \ ... \ T_n = C\ T_1'\ ... \ T_n', E\}, V, \sigma) \longrightarrow$
    $(\{T_1 = T_1',..., T_n = T_n', E\}, V, \sigma)$. 

      \

      In this case $n_{\mathrm{fvar}}, n_{\mathrm{rcomp}}, n_{\mathrm{svar}}$ do not change, $n_{\mathrm{rr}}$ becomes strictly smaller.

    \item Case. $(\{a\ T_1 \ ... \ T_n = C\ T_1'\ ... \ T_m', E\}, V, \sigma) \longrightarrow$
    $(\{T_i = C\ T_1'\ ... \ T_m', E \}, V, [\pi_n^i/a] \cdot \sigma)$.

      \

      In this case $n_{\mathrm{fvar}}$ does not change. If $T_i \equiv b \ T_1''\ ... \ T_l''$, where $b$ is a free variable,
      then $n_{\mathrm{rcomp}}$ does not change and $n_{\mathrm{svar}}$ becomes strictly smaller.
      Otherwise $n_{\mathrm{rcomp}}$ becomes strictly smaller.

  \item     $(\{a\ T_1 \ ... \ T_n = C\ T_1'\ ... \ T_m', E\}, V, \sigma) \longrightarrow$
    $(\{(b_1 \ T_1 \ ... \ T_n)  = T_1', ..., (b_m \ T_1 \ ... \ T_n) = T_m', E, \}, V, [\theta_n^m(C)/a]\cdot \sigma)$,
    where $b_1, ..., b_m \in \mathrm{FV}(\theta_n^m(C))$.

    \

    In this case $n_{\mathrm{fvar}}$ does not change. 
      But $n_{\mathrm{rcomp}}$ becomes strictly smaller.

  \end{itemize}
\end{proof}

\begin{theorem}[Soundness]
  \
  
\noindent If $(\{T_1 = T_1', ..., T_n = T_n' \}, V_1, \mathrm{\sigma_1}) \longrightarrow^* (\emptyset, V_2, \sigma_2 \cdot \sigma_1)$, $V_2 \# \mathrm{EV}(\mathrm{codom}(\sigma_2))$, $\mathrm{dom}(\sigma_2)\# \mathrm{codom}(\sigma_2)$ and $\mathrm{agree}(\sigma_2)$,
 then $\sigma_2 T_1 \equiv \sigma_2 T_1', ..., \sigma_2 T_n \equiv \sigma_2 T_n'$. 
\end{theorem}

\begin{proof}
  By induction on the length of $(\{T_1 = T_1', ..., T_n = T_n' \}, V_1, \mathrm{\sigma_1}) \longrightarrow^* (\emptyset, V_2, \sigma_2 \cdot \sigma_1)$.

  \begin{itemize}
  \item Base case. $(\{a = a \}, V, \sigma) \longrightarrow (\emptyset, V, \sigma)$.

    \
    
    This case is trivial.

  \item Base case. $(\{a = T \}, V, \sigma) \longrightarrow (\emptyset, V, [T/a] \cdot \sigma)$, if $a$ is first-order, $a \notin \mathrm{FV}(T)$
     and $T \not \equiv a$.

     \

     This case is straightforward.

   \item Step case: $(\{T = a, E\}, V, \sigma) \longrightarrow (\{a = T, E\}, V, \sigma) \longrightarrow^* (\emptyset, V', \sigma' \cdot \sigma)$, where $V' \# \mathrm{EV}(\mathrm{codom}(\sigma'))$, $\mathrm{codom}(\sigma')\# \mathrm{dom}(\sigma')$ and $\mathrm{agree}(\sigma')$.

     \

     By IH, we have $\sigma' a \equiv \sigma' T, \sigma E$ holds. Thus $\sigma' T \equiv \sigma' a, \sigma E$ holds. 

   \item Step case: $(\{C\ T_1'\ ... \ T_m' = a\ T_1 \ ... \ T_n, E\}, V, \sigma) \longrightarrow (\{a\ T_1 \ ... \ T_n = C\ T_1'\ ... \ T_m', E\}, V, \sigma) \longrightarrow^* (\emptyset, V', \sigma' \cdot \sigma)$, where $V' \# \mathrm{EV}(\mathrm{codom}(\sigma'))$,

     \noindent $\mathrm{codom}(\sigma')\# \mathrm{dom}(\sigma')$ and $\mathrm{agree}(\sigma')$.

      \

   This case is by straightforward induction. 

    \item Step case. $(\{\forall a . T = \forall b.T', E\}, V, \sigma) \longrightarrow (\{[a'/a]T = [a'/b]T', E\}, V\cup \{a'\}, \sigma)$ $\longrightarrow^* (\emptyset, V', \sigma' \cdot \sigma)$, where where $a'$ is a fresh eigenvariable, $V' \# \mathrm{EV}(\mathrm{codom}(\sigma'))$,  $\mathrm{codom}(\sigma')\# \mathrm{dom}(\sigma')$ and $\mathrm{agree}(\sigma')$.

    \
    
    By IH, we know that $\sigma'[a'/a]T \equiv \sigma'[a'/b]T', \sigma' E$ holds. Since $a' \in V'$ and $V' \# \mathrm{EV}(\mathrm{codom}(\sigma'))$, we have $\forall a' . \sigma'[a'/a]T \equiv \forall a' . \sigma'[a'/b]T'$. Thus $\sigma' (\forall a . T) \equiv \sigma' (\forall b . T')$.

    \item Step case. 

    $(\{C\ T_1 \ ... \ T_n = C\ T_1'\ ... \ T_n', E\}, V, \sigma) \longrightarrow$
    $(\{T_1 = T_1',..., T_n = T_n', E\}, V, \sigma)\longrightarrow^* (\emptyset, V', \sigma' \cdot \sigma)$, where $V' \# \mathrm{EV}(\mathrm{codom}(\sigma'))$, $\mathrm{codom}(\sigma')\# \mathrm{dom}(\sigma')$ and $\mathrm{agree}(\sigma')$.

      \

      This case is by straightforward induction. 

    \item Step case. $(\{a\ T_1 \ ... \ T_n = C\ T_1'\ ... \ T_m', E\}, V, \sigma) \longrightarrow$
      $(\{T_i = C\ T_1'\ ... \ T_m', E \}, V, [\pi_n^i/a] \cdot \sigma)\longrightarrow^* (\emptyset, V', \sigma'[\pi_n^i/a] \cdot \sigma)$, where $V' \# \mathrm{EV}(\mathrm{codom}(\sigma'[\pi_n^i/a]))$,

      $\mathrm{codom}(\sigma'[\pi_n^i/a])\# \mathrm{dom}(\sigma'[\pi_n^i/a])$ and $\mathrm{agree}(\sigma'[\pi_n^i/a])$.

      \

      We know that $V' \# \mathrm{EV}(\mathrm{codom}(\sigma'))$, $\mathrm{codom}(\sigma')\# \mathrm{dom}(\sigma')$ and $\mathrm{agree}(\sigma')$.
      By IH, we have $\sigma' T_i \equiv \sigma' (C\ T_1'\ ... \ T_m'), \sigma' E$. We
      need to show $\sigma'[\pi_n^i/a](a\ T_1 \ ... \ T_n) \equiv \sigma' [\pi_n^i/a] (C\ T_1'\ ... \ T_m'), \sigma'[\pi_n^i/a] E$, which is the case because $\sigma' [\pi_n^i/a] = [\pi_n^i/a]\sigma'$ and $\mathrm{codom}(\sigma'[\pi_n^i/a])\# \mathrm{dom}(\sigma'[\pi_n^i/a])$.

  \item Step case. $(\{a\ T_1 \ ... \ T_n = C\ T_1'\ ... \ T_m', E\}, V, \sigma) \longrightarrow$
    $(\{(b_1 \ T_1 \ ... \ T_n)  = T_1', ..., (b_m \ T_1 \ ... \ T_n) = T_m', E, \}, V, [\theta_n^m(C)/a]\cdot \sigma)\longrightarrow^* (\emptyset, V', \sigma'[\theta_n^m(C)/a] \cdot \sigma)$, where $V' \# \mathrm{EV}(\mathrm{codom}(\sigma'[\theta_n^m(C)/a]))$,
      $\mathrm{codom}(\sigma'[\theta_n^m(C)/a])\# \mathrm{dom}(\sigma'[\theta_n^m(C)/a])$, $\mathrm{agree}(\sigma'[\theta_n^m(C)/a])$ and  $b_1, ..., b_m \in \mathrm{FV}(\theta_n^m(C))$.

    \

   We know that $V' \# \mathrm{EV}(\mathrm{codom}(\sigma'))$,
   $\mathrm{codom}(\sigma')\# \mathrm{dom}(\sigma')$, $\mathrm{agree}(\sigma')$.
   By IH, we have $\sigma' (b_1 \ T_1 \ ... \ T_n)  \equiv \sigma' T_1', ..., \sigma' (b_m \ T_1 \ ... \ T_n) \equiv \sigma' T_m', \sigma' E$. This imples that $\sigma'[\theta_n^m(C)/a] (b_1 \ T_1 \ ... \ T_n)  \equiv \sigma'[\theta_n^m(C)/a] T_1', ..., \sigma'[\theta_n^m(C)/a] (b_m \ T_1 \ ... \ T_n) \equiv \sigma'[\theta_n^m(C)/a] T_m', \sigma'[\theta_n^m(C)/a] E$ because $\mathrm{codom}(\sigma'[\theta_n^m(C)/a])\# \mathrm{dom}(\sigma'[\theta_n^m(C)/a])$. So $\sigma'[\theta_n^m(C)/a]( a\ T_1 \ ... \ T_n ) \equiv \sigma'[\theta_n^m(C)/a](C\ T_1'\ ... \ T_m')$.

   \item The rest of the cases are straightforward.
  \end{itemize}
\end{proof}

\section{Soundness of the type checking algorithm}
\label{sound:type}
\begin{lemma}[Typing closed under substitution]
  \label{type:subst}
  \
  
  \noindent If $\Gamma \vdash e : T$, 
  then $\sigma \Gamma \vdash \sigma e : \sigma T$.
\end{lemma}
\begin{proof}
  By induction on the derivation of $\Gamma \vdash e : T$.
\end{proof}

We write $\mathrm{min}(L)$ to mean the minimal scope value in $L$, if $L$ is empty,
then we pick an arbitrary large number for $\mathrm{min}(L)$. We write $\mathrm{fst}(L)$
to mean all the variables in $L$. 

\begin{lemma}[Scope check composition]
  \label{compose}
  Suppose $\mathrm{FV}(\mathrm{codom}(\sigma_2)) \cap \mathrm{dom}(\sigma_1) = \emptyset$, $\mathrm{dom}(\sigma_1) \# \mathrm{codom}(\sigma_1)$ and $\mathrm{dom}(\sigma_2) \# \mathrm{codom}(\sigma_2)$. Moreover, let $L, L'$ be scope environments, where $\mathrm{fst}(L) \# \mathrm{fst}(L')$
  and all the variables in $L'$ are free variables. 

\noindent If $\mathsf{Scope}(L, \sigma_1)$ and $\mathsf{Scope}(\sigma_1(L + L'), \sigma_2)$, then $\mathsf{Scope}(L, \sigma_2 \cdot \sigma_1)$. 
\end{lemma}
\begin{proof}
  \begin{itemize}
  \item Case. $b \in \mathrm{dom}(\sigma_2 \cdot \sigma_1) - \mathrm{fst}(L)$. In
    this case there is nothing to check.

  \item Case. $b \in \mathrm{dom}(\sigma_2 \cdot \sigma_1) \cap \mathrm{fst}(L)$.

    We need to show
    for any $a \in \mathrm{EV}(\sigma_2 (\mathrm{FV}(\sigma_1 b))) \cup \mathrm{EV}(\sigma_1 b)$, $L(a) < L(b)$.
    Suppose $b \in \mathrm{dom}(\sigma_1)$. Since we
    know $\mathsf{Scope}(L, \sigma_1)$, so if $a \in \mathrm{EV}(\sigma_1 b)$, we have $L(a) < L(b)$. Let $u \in \mathrm{FV}(\sigma_1 b)$. 
    \begin{itemize}
    \item Suppose $u \notin \mathrm{dom}(\sigma_2)$. There is nothing to check. 
    \item Suppose $u \in \mathrm{dom}(\sigma_2)$. By definition of $\sigma_1(L+L')$, we have
      $u \in \mathrm{fst}(\sigma_1(L+L'))$. Note that $\sigma_1(L + L')$ does not change the scope values of the eigenvariables in $L$. Thus $\mathsf{Scope}(\sigma_1(L + L'), \sigma_2)$ implies that for any $q \in \mathrm{EV}(\sigma_2 u)$, we have $(\sigma_1(L + L'))(q) = L(q) < (\sigma_1(L + L'))(u) \leq L(b)$. Note that $(\sigma_1(L + L'))(u) \leq L(b)$ is by the definition of $\sigma_1(L + L')$.
      

    \end{itemize}

    Suppose $b \in \mathrm{dom}(\sigma_2)$, we just need to show
    for any $q \in \mathrm{EV}(\sigma_2 b)$, $L(q) < L(b)$. Since $\mathsf{Scope}(\sigma_1(L + L'), \sigma_2)$, we know $(\sigma_1(L + L'))(q) = L(q) < (\sigma_1(L + L'))(b) \leq L(b)$. 


  \end{itemize}
\end{proof}

\begin{lemma}[Scope check invariant]
  \label{sc:inv}

  
  \begin{enumerate}
  \item
    If $(\{(L_1, \Gamma_1,e_1,T_1),..., (L_n, \Gamma_n,e_n,T_n)\},\sigma) \longrightarrow$

    $(\{(L_1', \Gamma_1',e_1',T_1'),..., (L_m', \Gamma_m',e_m',T_m')\}, \sigma'\cdot \sigma)$,
    then $\mathsf{Scope}(L_i, \sigma')$ for all $i$.
  \item If $(\{(L_1, \Gamma_1,e_1,T_1),..., (L_n, \Gamma_n,e_n,T_n)\},\sigma) \longrightarrow^*$

    $(\{(L_1', \Gamma_1',e_1',T_1'),..., (L_m', \Gamma_m',e_m',T_m')\}, \sigma'\cdot \sigma)$,
    then $\mathsf{Scope}(L_i, \sigma')$ for all $i$. 
  \end{enumerate}
\end{lemma}

\begin{proof}
  \begin{enumerate}
  \item By case reasoning on $\longrightarrow$.
  \item By induction on the length of $\longrightarrow$.
    \begin{itemize}
    \item Base case: By (1).
    \item Step case:
      
      $(\{(L, \Gamma, (x | c)\ e_1 \ ... \ e_n, T), (L_1, \Gamma''_1, e_1'', T_1''),..., (L_k, \Gamma''_k,e_k'', T_k'') \}, \sigma) \longrightarrow_a $

      $(\{(\sigma'L', \sigma'\Gamma, e_1, \sigma' T_1'), ..., (\sigma' L', \sigma'\Gamma, e_l, \sigma' T_l'), (\sigma' L', \sigma'\Gamma, e_{l+1}, \sigma' b_{1}), ...,$

      $\quad \quad (\sigma' L', \sigma'\Gamma, e_{n}, \sigma' b_{n-l}), (\sigma' L_1, \sigma' \Gamma''_1, e_1'', \sigma' T_1''),..., (\sigma' L_k, \sigma' \Gamma''_k,e_k'', \sigma' T_k'')\}, \sigma' \cdot \sigma) \longrightarrow^*$
      $(\Psi, \sigma'' \cdot \sigma' \cdot \sigma)$,

          where $ l \leq n > 0$, $(x | c) : \forall a_1 ... \forall a_k. T_1',..., T_l' \to T' \in \Gamma$ and $\mathsf{Scope}(L, \sigma')$. Here $\sigma', L'$ is defined by the following.

    \begin{itemize}

    \item $T' \sim_{\sigma'} (b_{1} ,..., b_{n-l} \to T)$, where $b_{1},..., b_{n-l}$ are fresh free variables. .
    \item $L' = L + [(a_i, n') | 1 \leq i \leq k] + [(b_j, n') | 1\leq b_j \leq n-l]$, where $n'= \mathrm{max}(L)+1$. Moreover, $a_1,..., a_k$ are fresh free variables.  
    \end{itemize}

    \
    
      By IH, we know that $\mathsf{Scope}(\sigma' L', \sigma'')$ and $\mathsf{Scope}(\sigma'L_i, \sigma'')$ for all $i$. By Lemma \ref{compose}, we have $\mathsf{Scope}(L, \sigma'' \cdot \sigma')$ and $\mathsf{Scope}(L_i, \sigma'' \cdot \sigma')$ for all $i$.

    \item The step case for $\longrightarrow_b$ is similar to $\longrightarrow_a$.

    \item Step case:
      \
      
          $(\{(L, \Gamma, \lambda x_1. ... \lambda x_n . e \ e_1' ... e_l', \forall a_1. ... \forall a_m . T), (L_1, \Gamma''_1, e_1'', T_1''),..., (L_k, \Gamma''_k,e_k'', T_k'') \}, \sigma)$

      $ \longrightarrow_i (\{([L, (a_1, n'+1),..., (a_m, n'+m)], [\Gamma, x_1 : T_1,..., x_n : T_n], e \ e_1' ... e_l',  T), $

      $\quad\quad\quad \quad (L_1,  \Gamma''_1, e_1'',  T_1''),..., ( L_k, \Gamma''_k,e_k'',  T_k'')\}, \sigma) \longrightarrow^*$
      $(\Psi, \sigma' \cdot \sigma)$, where $n, l \geq 0$, $m >0$, $a_1,..., a_m$ are fresh eigenvariables and $n' = \mathrm{max}(L)$.

      \
      
            By IH, we have $\mathsf{Scope}(L_i, \sigma')$ for all $i$ and $\mathsf{Scope}([L, (a_1, n'+1),..., (a_m, n'+m)], \sigma')$. Thus $\mathsf{Scope}(L, \sigma')$ as $a_1,..., a_m$ are fresh eigenvariables.

          \item Step case:
            \
            
            $(\{(L, \Gamma, (x | c), T), (L_1, \Gamma''_1, e_1'', T_1''),..., (L_k, \Gamma''_k,e_k'', T_k'')\}, \sigma) \longrightarrow_s $

            $(\{(\sigma'L', \sigma'\Gamma, \diamond, \diamond), (\sigma' L_1, \sigma' \Gamma''_1, e_1'', \sigma' T_1''),..., (\sigma' L_k, \sigma' \Gamma''_k,e_k'', \sigma'T_k'')\}, \sigma' \cdot \sigma)$
      $\longrightarrow^* (\Psi, \sigma'' \cdot \sigma' \cdot \sigma)$, where $(x|c) : T' \in \Gamma$ and $\mathsf{Scope}(L, \sigma')$, where $(x|c) : \forall a_1 ... \forall a_k . T'' \in \Gamma$. Here $\sigma', L'$ is defined by the followings.
    \begin{itemize}
    \item $T'' \sim_{\sigma'} T$
    \item $L' = L + [(a_i, n)\ |\ 0 < i \leq k, n = \mathrm{max}(L)+1]$. Note that $a_1,..., a_k$ are fresh free variables in $T''$. 
    \end{itemize}

    \
    
    By IH, we know that $\mathsf{Scope}(\sigma'L', \sigma'')$ and  $\mathsf{Scope}(\sigma'L_i, \sigma'')$ for all $i$. By Lemma \ref{compose}, we have $\mathsf{Scope}(L, \sigma'' \cdot \sigma')$ and $\mathsf{Scope}(L_i, \sigma'' \cdot \sigma')$ for all $i$. 

    \end{itemize}
  \end{enumerate}

\end{proof}

\begin{lemma}[Scope environment records free variables]
  \label{record}
Suppose we have the following transition. 
  
\noindent  {\footnotesize $(\{(L_1,\Gamma_1,e_1,T_1),..., (L_n,\Gamma_n,e_n,T_n)\},\sigma)\longrightarrow$
 $(\{(L_1', \Gamma_1',e_1',T_1'),..., (L_m', \Gamma_m',e_m',T_m')\}, \sigma'\sigma)$}.

  If $(\mathrm{FV}(\Gamma_i) \cup \mathrm{FV}(T_i)) \subseteq \mathrm{fst}(L_i)$ for $1 \leq i \leq n$, then $(\mathrm{FV}(\Gamma_j') \cup \mathrm{FV}(T_j')) \subseteq \mathrm{fst}(L_j')$ for $1 \leq j \leq m$. 
\end{lemma}

\begin{proof}
  By case analysis on Definition \ref{rsm} and the definition of $\sigma L$.
\end{proof}

In the following proof, we use $\sigma |_{\underline{a}}$ to mean restricting the domains of
$\sigma$ to be $\{a_1,..., a_k\}$ for some $k$. We use $\forall \underline{a}. T$ to
denote $\forall a_1. ... \forall a_k . T$. Furthermore, $e\ (\sigma \underline{a})$ means
$e\ \sigma a_1 \ ... \ \sigma a_k$. 
\begin{theorem}[Soundness]
  If $(\{(L_1, \Gamma_1,e_1,T_1),..., (L_n, \Gamma_n,e_n,T_n)\},\sigma) \longrightarrow^*$

  \noindent $(\{(L_1', \Gamma_1', \diamond, \diamond),..., (L_m', \Gamma_m',\diamond,\diamond)\}, \sigma'\cdot \sigma)$ and $(\mathrm{FV}(\Gamma_i) \cup \mathrm{FV}(T_i)) \subseteq \mathrm{fst}(L_i)$ for $1 \leq i \leq n$, then there exists $p_i$ such that $\sigma' \Gamma_i \vdash p_i : \sigma' T_i$ and $|p_i| = e_i$ for all $i$.

\end{theorem}
\begin{proof} By induction on the length of $(\{(L_1, \Gamma_1, e_1, T_1),..., (L_n, \Gamma_n, e_n, T_n)\},\sigma) \longrightarrow^* (\{(L_1', \Gamma_1', \diamond, \diamond),..., (L_m', \Gamma_m',\diamond,\diamond)\}, \sigma'\cdot \sigma)$.
  \begin{itemize}

  \item Base case: $(\{(L, \Gamma, (x | c), T)\}, \sigma) \longrightarrow_s (\{(\sigma'L', \sigma'\Gamma, \diamond, \diamond)\}, \sigma' \cdot \sigma)$

    if $(x|c) :  \forall a_1 ... \forall a_k . T'' \in \Gamma$ and $\mathsf{Scope}(L, \sigma')$. Here $\sigma', L'$ is defined by the followings.
    \begin{itemize}

      

      \item $T'' \sim_{\sigma'} T$
    \item  $L' = L + [(a_i, n)\ |\ 0 < i \leq k, n = \mathrm{max}(L)+1]$. Note that $a_1,..., a_k$ are the fresh free variables in $T''$. 
    \end{itemize}
      In this case, since $\Gamma \vdash (x | c) : \forall a_1 ... \forall a_n. T''$, by instantiation, we have
    $\Gamma \vdash (x | c)\ (\sigma' \underline{a}) : \sigma'T''$. By Lemma \ref{type:subst}, we have $\sigma' \Gamma \vdash (x | C)\ (\sigma' \underline{a}) :  \sigma' \sigma'(T'') \equiv \sigma' T'' \equiv \sigma' T$. 

    \item Step case: 

      $(\{(L, \Gamma, (x|c)\ e_1 \ ... \ e_n, T), (L_1, \Gamma''_1, e_1'', T_1''),..., (L_k, \Gamma''_k,e_k'', T_k'') \}, \sigma) \longrightarrow_a $

      $(\{(\sigma'L', \sigma'\Gamma, e_1, \sigma' T_1'), ..., (\sigma' L', \sigma'\Gamma, e_l, \sigma' T_l'), $

      \quad\quad$(\sigma' L', \sigma'\Gamma, e_{l+1}, \sigma' b_{1}), ..., (\sigma' L', \sigma'\Gamma, e_{n}, \sigma' b_{n-l}),$

      $\quad \quad \quad (\sigma' L_1, \sigma' \Gamma''_1, e_1'', \sigma' T_1''),..., (\sigma' L_k, \sigma' \Gamma''_k,e_k'', \sigma' T_k'')\}, \sigma' \cdot \sigma) \longrightarrow^*$
      
      $(\{(L_1', \Gamma_1', \diamond, \diamond),..., (L_m', \Gamma_m',\diamond,\diamond)\}, \sigma'' \cdot \sigma' \cdot \sigma)$, 

      \noindent where $l \leq n > 0$, $(x | c) : \forall a_1 ... \forall a_k. T_1',..., T_l' \to T' \in \Gamma$ and $\mathsf{Scope}(L, \sigma')$. Here $\sigma', L'$ is defined by the following.
      \begin{itemize}
    \item $T' \sim_{\sigma'} (b_{1} ,..., b_{n-l} \to T)$, where $b_{1},..., b_{n-l}$ are fresh free variables.
    \item $n'= \mathrm{max}(L)+1$, $L' = L + [(a_i, n') | 1 \leq i \leq k] + [(b_j, n') | 1\leq b_j \leq n-l]$. Note that $a_1,..., a_k$ are fresh free variables in $T_1',..., T_l' \to T'$. 
      \end{itemize}

      By IH, we know that $\sigma'' \sigma'\Gamma \vdash p_1 : \sigma''\sigma' T_1', ..., \sigma''' \sigma' \Gamma \vdash p_l' : \sigma''\sigma' T_l', \sigma''\sigma' \Gamma \vdash p_{l+1} : \sigma'' \sigma' b_{1},..., \sigma'' \sigma'\Gamma \vdash p_{n} : \sigma'' \sigma' b_{n-l}$, where $|p_1| = e_1,..., |p_l| = e_l, |p_{l+1}| = e_{l+1}, ..., |p_n| = e_n$.

      Furthermore, $\sigma'' \sigma' \Gamma''_1 \vdash  p_1' :  \sigma'' \sigma' T_1'',..., \sigma'' \sigma' \Gamma''_k \vdash p_k' : \sigma'' \sigma' T_k''$, $|p_1'| = e_1'',..., |p_k'| = e_k''$. 

    We know that $(x|c) : \forall \underline{a}. T_1',..., T_l' \to T' \in \Gamma$ and
    $T' \sim_{\sigma'} (b_{1} ,..., b_{n-l} \to T)$. By the instantiation typing rule, we have $\Gamma \vdash (x|c) \ (\sigma' \underline{a}) : \sigma'|_{\underline{a}} (T_1',...,  T_l' \to T')$. By Lemma \ref{type:subst}, we have

    \begin{center}
      $\sigma'\Gamma \vdash (x|c) \ (\sigma' \underline{a}) : \sigma' \sigma'|_{\underline{a}} (T_1',...,  T_l' \to T') \equiv \sigma'(T_1',..., T_l',  b_1,...,  b_{n-l} \to T) \equiv \sigma' T_1',..., \sigma' T_l', \sigma' b_1,..., \sigma' b_{n-l} \to \sigma' T'$. 
    \end{center}
    So by Lemma \ref{type:subst}, we have

    $\sigma'' \sigma' \Gamma \vdash (x|c)\ (\sigma''\sigma' \underline{a}): \sigma''\sigma' T_1',..., \sigma''\sigma' T_l', \sigma''\sigma' b_1,..., \sigma''\sigma' b_{n-l} \to \sigma'' \sigma' T'$.

    Finally, we have $\sigma''' \sigma''\Gamma \vdash ((x|c) \ (\sigma''\sigma' \underline{a}))\ p_1 \ ... \ p_n :\sigma'' \sigma' T'$

    and $|((x|c)\ (\sigma''\sigma' \underline{a}))\ p_1 \ ... \ p_n| = (x|c)\ e_1 \ ... \ e_n$. 

  \item Step case:
    
    $(\{(L, \Gamma, (x|c)\ e_1 \ ... \ e_n, T), (L_1, \Gamma''_1, e_1'', T_1''),..., (L_k, \Gamma''_k,e_k'', T_k'') \}, \sigma) \longrightarrow_b$

    $(\{(\sigma'L', \sigma'\Gamma, e_1, \sigma' T_1'), ..., (\sigma' L', \sigma'\Gamma, e_{n}, \sigma' T_n'),$

    \quad \quad $(\sigma' L_1, \sigma' \Gamma''_1, e_1'', \sigma' T_1''),..., (\sigma' L_k, \sigma' \Gamma''_k,e_k'', \sigma' T_k'')\},$ $ \sigma' \cdot \sigma) \longrightarrow^*$
    
    $(\{(L_1', \Gamma_1', \diamond, \diamond),..., (L_m', \Gamma_m',\diamond,\diamond)\}, \sigma'' \cdot \sigma' \cdot \sigma)$,

    where $0< n < l$, $(x | c) : \forall a_1 ... \forall a_k. T_1',..., T_l' \to T' \in \Gamma$ and $\mathsf{Scope}(L, \sigma')$. Here $\sigma', L'$ is defined by the following.
    
      \begin{itemize}
    \item $(T_{n+1}',..., T_l' \to T') \sim_{\sigma'} T$
    \item $L' = L + [(a_i, n') \ | 1 \leq i \leq k, n'= \mathrm{max}(L)+1 ]$. Note that $a_1,..., a_k$ are fresh free variables in $T_1',..., T_l' \to T'$. 
      \end{itemize}

    By IH, we know that $\sigma''\sigma' \Gamma \vdash p_1 : \sigma'' \sigma' T_1', ..., \sigma'' \sigma'\Gamma \vdash p_n : \sigma'' \sigma' T_n',  \sigma'' \sigma'\Gamma''_1 \vdash p_1' : \sigma'' \sigma'T_1'',..., \sigma'' \sigma' \Gamma''_k \vdash p_k' : \sigma'' \sigma' T_k''$, where $|p_1| = e_1, ..., |p_n| = e_n, |p_1'| = e_1'',..., |p_k'|= e_k''$. 

    We know that
    $(x|c) : \forall \underline{a}. T_1',..., T_l' \to T' \in \Gamma$ and $(T_{n+1}',..., T_l' \to T') \sim_{\sigma'} T$. By the instantiation rule, we know that $\Gamma \vdash (x|c)\ (\sigma' \underline{a}) : \sigma'|_{\underline{a}}(T_1',..., T_l' \to T')$. By Lemma \ref{type:subst}, we know that $\sigma'\Gamma \vdash (x|c)\ (\sigma' \underline{a}) : \sigma' \sigma'|_{\underline{a}}(T_1',..., T_l' \to T') \equiv \sigma' T_1',..., \sigma' T_n' \to \sigma' T$. 

By Lemma \ref{type:subst}, we have $\sigma''\sigma' \Gamma \vdash (x|c)\ (\sigma''\sigma' \underline{a}) : \sigma''\sigma' T_1',..., \sigma''\sigma' T_n' \to \sigma'' \sigma' T$. Finally, $\sigma''\sigma' \Gamma \vdash  (x|c)\ (\sigma''\sigma' \underline{a})\ p_1 \ ... \ p_n : \sigma''\sigma' T$, where $| (x|c)\ (\sigma''\sigma' \underline{a})\ p_1 \ ... \ p_n| = (x|c) \ e_1 \ ... \ e_n$.

\item Step case:

  $(\{(L, \Gamma, \lambda x_1.... \lambda x_n. e\ e'_1 ... e_l', \forall a_1. ...\forall a_m . T_1,..., T_n \to T),$

  \quad \quad$ (L_1, \Gamma''_1, e_1'', T_1''),..., (L_k, \Gamma''_k,e_k'', T_k'')\}, \sigma) \longrightarrow_i$

  $(\{([L, (a_1, n'+1),..., (a_m, n'+m)], [\Gamma, x_1 : T_1,..., x_n : T_n], e \ e'_1 ... e_l', T), $

  \quad $(L_1, \Gamma''_1, e_1'', T_1''),..., (L_k, \Gamma''_k, e_k'', T_k'')\}, \sigma) \longrightarrow^*$

  $(\{(L_1', \Gamma_1', \diamond, \diamond),..., (L_m', \Gamma_l',\diamond,\diamond)\}, \sigma' \cdot \sigma)$, where $n, l \geq 0, m > 0$, $a_1,..., a_m$ are fresh eigenvariables and $n' = \mathrm{max}(L)$.

  By IH, we know that
    $\sigma' \Gamma, x_1 : \sigma' T_1,..., x_n : \sigma' T_n \vdash p : \sigma' T,  \sigma' \Gamma''_1 \vdash p_1 : \sigma' T_1'',..., \sigma'\Gamma''_k \vdash p_k : \sigma' T_k''$, where $|p| = e\ e_1' ... \ e_l', |p_1| = e_1'',..., |p_k| = e_k''$. Thus $\sigma' \Gamma \vdash \lambda x_1...\lambda x_n . p : \sigma' T_1,..., \sigma' T_n \to \sigma' T$.

  By Lemma \ref{sc:inv}, we have $\mathsf{Scope}([L, (a_1, n'+1),..., (a_m, n'+m)], \sigma')$.
  Let $F = \mathrm{FV}(\Gamma) \cup \mathrm{FV}(\forall a_1. ...\forall a_m . T_1,..., T_n \to T)$, by Lemma \ref{record}, $F \subseteq \mathrm{fst}(L)$. Thus $\mathrm{EV}(\mathrm{codom}(\sigma'|_F)) \cap \{a_1,..., a_m\} = \emptyset$. Thus $\{a_1,...,a_m\} \cap \mathrm{FV}(\sigma'\Gamma) = \emptyset$, by the abstraction rule, we have $\sigma'\Gamma \vdash \lambda a_1. ... \lambda a_m. \lambda x_1...\lambda x_n . p : \forall \underline{a}.\sigma' T_1,..., \sigma' T_n \to \sigma' T$. Furthermore, $\forall \underline{a}.\sigma' T_1,..., \sigma' T_n \to \sigma' T = \sigma'(\forall a_1....\forall a_m. T_1,..., T_n \to T)$ (because $\mathrm{EV}(\mathrm{codom}(\sigma'|_F)) \cap \{a_1,..., a_m\} = \emptyset$) and
    $|\lambda a_1. ... \lambda a_m. \lambda x_1...\lambda x_n . p| = \lambda x_1...\lambda x_n . e\ e'_1 ... e_l'$.

    \item The final step case $\longrightarrow_s$ is proved similarly. 





  \end{itemize}
\end{proof}

\section{Type checking examples}
\label{examples}
In this section we will give two main examples of show we can benefit from the support of impredicative and second-order types. The prototype type checker is available at \url{https://github.com/fermat/higher-rank}. There are more examples in the \texttt{/examples} directory in the prototype.

Firstly, to familiar with the output of the type checker, consider the following program (\texttt{examples/ex1.hr}).
\begin{verbatim}
f :: (forall a . a) -> forall a . a
f = \ x -> x x

id :: forall a . a -> a
id x = x
\end{verbatim}

\noindent Our type checker will output the following annotated program, which is then checked
by a separated proof checker.

\begin{verbatim}
f :: (forall a . a) -> (forall a . a) =
  \ (x :: forall a . a) . x @((forall a . a) -> (forall a . a)) x 

id :: forall a . a -> a =
  \\ a0# . \ (x :: a0#) . x 
\end{verbatim}

\noindent Note that in the output program,
we use \texttt{\string\ x . e} to denote the usual lambda abstraction and
\texttt{\string\\ x . e} to denote type abstraction. Every lambda abstraction
is annotated with its type, since kind inference is decidable, we do not
annotate the type abstraction. The machine generated type variables are postfixed
with $\#$ symbol and we use $@ T$ to denote type instantiation. 
\subsection{Bird and Paterson's program}
Without the support of second-order types, Bird and Paterson have to write the
following program.

\begin{verbatim}
gfoldT :: forall m n b . 
              (forall a. m a -> n a) ->
                (forall a. Pair (n a) -> n a) ->
                  (forall a. n (Incr a) -> n a) ->
                    (forall a. Incr (m a) -> m (Incr a)) ->
                     Term (m b) -> n b
gfoldT v a l k (Var x) = v x
gfoldT v a l k (App p) = (a . mapP (gfoldT v a l k)) p
gfoldT v a l k (Lam t) = (l . gfoldT v a l k . mapT k) t

kfoldT :: forall c b . (c -> b) ->
                        (Pair b -> b) ->
                          (b -> b) ->
                           (Incr c -> c) ->
                             Term c -> b
kfoldT v a l k (Var x) = v x
kfoldT v a l k (App p) = (a . mapP (kfoldT v a l k)) p
kfoldT v a l k (Lam t) = (l . kfoldT v a l k . mapT k) t

showT :: Term String -> String
showT = kfoldT id showP (\ x -> lambda ++ x) showI
\end{verbatim}
\noindent Note that \texttt{kfoldT} has the exact same definition as \texttt{gfoldT},
but with a more specific type, i.e. if we instantiate \texttt{m} with \texttt{\string\ x . c}
and instantiate \texttt{n} with \texttt{\string\ x . b} in the type of \texttt{gfoldT}, then we
get the type for \texttt{kfoldT}. Moreover, when \texttt{kfoldT} is used by \texttt{showT}, \texttt{c} and \texttt{b} are both instantiated with \texttt{String}.

With the support of the second-order type, we now can write the following program. 

\begin{verbatim}
gfoldT :: forall m n b . 
              (forall a. m a -> n a) ->
                (forall a. Pair (n a) -> n a) ->
                  (forall a. n (Incr a) -> n a) ->
                    (forall a. Incr (m a) -> m (Incr a)) ->
                     Term (m b) -> n b
gfoldT v a l k (Var x) = v x
gfoldT v a l k (App p) = (a . mapP (gfoldT v a l k)) p
gfoldT v a l k (Lam t) = (l . gfoldT v a l k . mapT k) t

showT :: Term String -> String
showT = gfoldT id showP (\ x -> lambda ++ x) showI
\end{verbatim}

\noindent Note that when \texttt{gfoldT} is used by \texttt{showT}, both \texttt{m} and
\texttt{n} are instantiated with \texttt{\string\ x . String}. Please see the file \texttt{examples/bird.hr} for the whole program. 
    
\subsection{Stump's impredicative merge sort}
It is well-known that $\mathbf{F}_\omega$ can support impredicative Church encoding. However, 
programming with Church encoding directly in $\mathbf{F}_\omega$ can be a bit cumbersome due to
the amount of type annotations that required.
In Haskell, we can avoid a lot of these annotations by using data constructor to perform explicit type conversion. The following program (from Jones \cite{jones1997first}) is typically what one would write in Haskell. Note that the extra data constructors \texttt{L} is used
to explicitly convert back and forth between type \texttt{List a} and type \texttt{(a -> b -> b) -> b -> b}.
 
\begin{verbatim}
data List a = L ((a -> b -> b) -> b -> b)

fold :: forall a b . List a -> (a -> b -> b) -> b -> b
fold (L f) = f

nil :: forall a . List a
nil = L (\c n -> n)

cons :: forall a . a -> List a -> List a
cons x xs = L (\c n -> c x (fold xs c n))

hd :: forall a . List a -> a
hd l = fold l (\x xs -> x) (error "hd []")

tl :: forall a . List a -> List a
tl l = fst (fold l c n)
  where c x (l,t) = (t, cons x t)
        n = (error "tl []", nil)  
\end{verbatim}

With the support of impredicative types, we can program with Church encoding without
explicit type conversion, so no extra data constructor is needed\footnote{ Please see \texttt{examples/church.hr} for more simple examples of Church encodings.}. 

\begin{verbatim}
type List :: * -> * = \ a . forall x . (a -> x -> x) -> x -> x

fold :: forall a b . List a -> (a -> b -> b) -> b -> b
fold l f n = l f n

nil :: forall a . List a
nil = \ c n -> n

cons :: forall a . a -> List a -> List a
cons = \ a as c n -> c a (as c n)

head :: forall a . List a -> a
head l = l (\ a r -> a) undefined

tail :: forall a . List a -> (List a) 
tail l =  fst (l (\ a r -> (snd r, cons a (snd r))) (nil, nil))
\end{verbatim}

Stump implemented a mini-language called \texttt{fore} for $\mathbf{F}_\omega$ in order to program with different lambda-encoding schemes \cite{stump2016efficiency}. The following is a merge sort program in \texttt{fore} using Church encoding (from \cite{stump2016efficiency}).   

\newpage
\begin{verbatim}
merge : forall A : *, (A -> A -> Bool) -> List A -> List A -> List A 
= \ A : *, \ cmp : A -> A -> Bool, \ la : List A,
  la (List A -> List A -> List A)
     (\ a : A , \ outer : List A -> List A -> List A,
        \ la : List A, \ lb : List A,
          head A la (List A)
           (\ ha : A,
             lb (List A -> List A)
              (\ b : A, \ inner : List A -> List A, 
               \ lb : List A,
                   head A lb (List A) 
                    (\ hb : A,
                       cmp ha hb (List A)
                         (Cons A ha (outer (tail A la) lb))
                         (Cons A hb (inner (tail A lb))))
                     la)
              (\ lb : List A, la)
               lb) 
            lb)
      (\ la : List A ,\ lb : List A, lb) 
      la.  
\end{verbatim}

 We can see there are a lot of type annotations required.
With the support of impredicative types, the above program can be simplified to the
following, where the only type annotation required is at the top level. Please see \texttt{examples/church-braun.hr} for the full definitions.  

\begin{verbatim}
merge :: forall a . (a -> a -> Bool) -> List a -> List a -> List a
merge f la = 
  la (\ laa lbb -> lbb)
     (\ a outer la' lb ->
           caseMaybe (head la')
              (\ ha ->
                  lb (\ lb'' -> la')
                     (\ b inner lb' ->
                        caseMaybe (head lb')
                           (\ hb ->
                               if (f ha hb) 
                                  (cons ha (outer (tail la') lb'))
                                  (cons hb (inner (tail lb'))))
                            la')
                     lb)
               lb)
     la 
\end{verbatim}
\end{document}